%% file: sc-intro.tex
\documentclass[acmsmall]{acmart}
\acmJournal{PACMPL}
\acmVolume{1}
\acmNumber{ICFP}
\acmArticle{1}
\acmYear{2024}
\acmMonth{1}
\acmDOI{}
\startPage{1}

\setcopyright{none}

\usepackage{booktabs}
\usepackage{subcaption}
\usepackage{amsmath}
\usepackage[all,cmtip]{xy}
\usepackage[nameinlink,noabbrev,capitalise]{cleveref}
\usepackage{csquotes}
\usepackage{booktabs}
\usepackage{microtype}
\usepackage{stmaryrd}
\usepackage{graphicx}
\usepackage{enumitem}
\usepackage{mathtools}
\usepackage{courier}
\usepackage{bussproofs}
\usepackage{tikz-cd}
\usepackage{diagbox}
\usepackage{afterpage}
\usepackage[normalem]{ulem}
\usepackage{dashbox}
\usepackage{mdframed}
\usepackage[mathscr]{euscript}
\DeclareMathAlphabet\mathbfscr{U}{eus}{b}{n}
\usepackage{colortbl}
\usepackage{anyfontsize}
\usepackage{listings}
\usepackage{alltt} 
\usepackage{multicol} 
\usepackage{scalerel} 
\usepackage{turnstile} 
\usepackage{cmll}
\usepackage{wrapfig}

\usepackage[]{todonotes}

\usepackage{syntax}

\input{macros/macros.tex}

\usepackage{xcolor}

\theoremstyle{remark}
\newtheorem{remark}{Remark}
\newtheorem{frameddefinition}{Definition}[section]
\usepackage{mdframed} 
\surroundwithmdframed{frameddefinition}

\begin{document}

\title{Grokking the Sequent Calculus (Functional Pearl)}



\author{David Binder}
\orcid{0000-0003-1272-0972}
\affiliation{
  \department{Department of Computer Science}
  \institution{University of Tübingen}
  \streetaddress{Sand 14}
  \city{Tübingen}
  \postcode{72076}
  \country{Germany}
}
\email{david.binder@uni-tuebingen.de}

\author{Marco Tzschentke}
\orcid{0009-0004-8834-2984}
\affiliation{
  \department{Department of Computer Science}
  \institution{University of Tübingen}
  \streetaddress{Sand 14}
  \city{Tübingen}
  \postcode{72076}
  \country{Germany}
}
\email{marco.tzschentke@uni-tuebingen.de}

\author{Marius Müller}
\orcid{0000-0002-0260-6298}
\affiliation{
  \department{Department of Computer Science}
  \institution{University of Tübingen}
  \streetaddress{Sand 14}
  \city{Tübingen}
  \postcode{72076}
  \country{Germany}
}
\email{mari.mueller@uni-tuebingen.de}

\author{Klaus Ostermann}
\orcid{0000-0001-5294-5506}
\affiliation{
  \department{Department of Computer Science}
  \institution{University of Tübingen}
  \streetaddress{Sand 14}
  \city{Tübingen}
  \postcode{72076}
  \country{Germany}
}
\email{klaus.ostermann@uni-tuebingen.de}

\begin{abstract}
    The sequent calculus is a proof system which was designed as a more symmetric alternative to natural deduction.
    The \lambdamumu\ is a term assignment system for the sequent calculus and a great foundation for compiler intermediate languages due to its first-class representation of evaluation contexts.
    Unfortunately, only experts of the sequent calculus can appreciate its beauty.
    To remedy this, we present the first introduction to the \lambdamumu\ which is not directed at type theorists or logicians but at compiler hackers and programming-language enthusiasts.
    We do this by writing a compiler from a small but interesting surface language to the \lambdamumu\ as a compiler intermediate language.
\end{abstract}

\maketitle

\section{Introduction}
\label{sec:intro}
\input{sections/introduction.tex}

\section{Translating To Sequent Calculus}
\label{sec:syntax}
\input{sections/syntax.tex}

\section{Evaluation Within a Context}
\label{sec:focusing}
\input{sections/focusing.tex}

\section{Typing Rules}
\label{sec:typing}
\input{sections/typing.tex}

\section{Insights}
\label{sec:insights}
\input{sections/insights.tex}

\section{Related Work}
\label{sec:related-work}
\input{sections/related-work.tex}

\section{Conclusion}
\label{sec:conclusion}
\input{sections/conclusion.tex}

\section*{Data Availability Statement}
This paper is accompanied by an implementation in Haskell.
Upon acceptance and publication of this article, the implementation will be made available permanently using Zenodo.

\section*{Acknowledgments}

We would like to thank the anonymous reviewers and Paul Downen for their feedback which helped us greatly in improving the final presentation of the paper.

\appendix

\section{The Relationship to the Sequent Calculus}
\label{sec:appendix:sequent-calculus-relationship}
\input{sections/sequent-calculus.tex}

\section{Typing Rules for \surfacelang}
\label{sec:appendix:typing-surfacelang}
\input{sections/surface-typing.tex}

\section{Operational Semantics of $\mathbf{label}/\mathbf{goto}$}
\label{sec:appendix:semantics-label}
\input{sections/semantics-label.tex}

\bibliography{bibliography/bibliography.bib, bibliography/ownpublications.bib}

\end{document}

%% file: macros/macros.tex
\newcommand{\lambdamumu}{$\lambda\mu\tilde\mu$-calculus}
\newcommand{\surfacelang}{\textbf{Fun}}
\newcommand{\targetlang}{\textbf{Core}}

\newcommand{\natlit}[1]{\ensuremath{\ulcorner#1\urcorner}}
\newcommand{\tyint}{\mathbf{Int}}

\newcommand{\ifzero}[3]{\mathbf{ifz}(#1,#2,#3)}
\newcommand{\case}[2]{\mathbf{case}\ #1\ \mathbf{of}\ \{ #2 \}}
\newcommand{\cocase}[1]{\mathbf{cocase}\ \{ #1 \}}
\newcommand{\letin}[3]{\mathbf{let}\ #1 = #2\ \mathbf{in}\ #3}
\newcommand{\xvdash}[1]{\vdash_{#1}}
\newcommand{\jump}[2]{\mathbf{goto}(#1;#2)}
\newcommand{\labelterm}[2]{\mathbf{label}\ #1\ \{#2 \}}
\newcommand{\ite}[3]{\ensuremath{\mathbf{if}\ #1\ \mathbf{then}\ #2\ \mathbf{else}\ #3}}

\newcommand{\reducesto}{\triangleright}

\newcommand{\vartop}{\bigstar}
\newcommand{\valueof}[1]{\mathfrak{#1}}

\newcommand{\cut}[2]{\ensuremath{\langle #1 \mid #2 \rangle}}
\newcommand{\SCcase}[1]{\ensuremath{\mathbf{case}\ \{ #1 \}}}
\newcommand{\SCcocase}[1]{\ensuremath{\mathbf{cocase}\ \{ #1 \}}}

\newcommand{\wrap}[1]{:^{\text{\tiny#1}}}
\newcommand{\prd}{\wrap{prd}}
\newcommand{\cnt}{\wrap{cns}}
\newcommand{\wellformed}[1]{#1\ \textsc{Ok}}

\newcommand{\translate}[1]{\llbracket #1 \rrbracket}
\newcommand{\focus}[1]{\ensuremath{\mathcal{F}(#1)}}
\newcommand{\fresh}[1]{\ensuremath{(#1\text{ fresh})}}
\newcommand{\novalue}[1]{\ensuremath{(#1\text{ not a value})}}

%% file: sections/introduction.tex
Suppose you have just implemented your own small functional language.
To test it, you write the following function which multiplies all the numbers contained in a list:
\begin{align*}
  \mathbf{def}\ \text{mult}(l) \coloneq \case{l}{\, \text{Nil} \Rightarrow 1, \text{Cons}(x,xs) \Rightarrow x * \text{mult}(xs)\,}
\end{align*}
What bugs you about this implementation is that you know an obvious optimization:
The function should directly return zero if it encounters a zero in the list.
There are many ways to achieve this, but you choose to extend your language with labeled expressions and a goto instruction.
This allows you to write the optimized version:
\begin{align*}
  &\mathbf{def}\ \text{mult}(l) \coloneq \labelterm{\alpha}{\, \text{mult'}(l;\alpha)\,} \\
  &\mathbf{def}\ \text{mult'}(l;\alpha) \coloneq \case{l}{\, \text{Nil} \Rightarrow 1, \text{Cons}(x,xs) \Rightarrow \ifzero{x}{\jump{0}{\alpha}}{x * \text{mult'}(xs; \alpha)}\,}
\end{align*}
You used $\labelterm{\alpha}{\text{mult'}(l;\alpha)}$ to introduce a label $\alpha$ around the call to the helper function $\text{mult'}$ which takes this label as an additional argument (we use $;$ to separate the label argument from the other arguments), and $\jump{0}{\alpha}$ to jump to this label $\alpha$ with the expression $0$ in the recursive helper function.
But since your language now has control effects, you need to reconsider how you want to compile and optimize programs.
In particular, you have to decide on an appropriate intermediate language which can express these control effects.
In this paper, we introduce you to one such intermediate language: the sequent-calculus-based \lambdamumu.
The result of compiling the efficient multiplication function to the \lambdamumu\ is:
\begin{align*}
  &\mathbf{def}\ \text{mult}(l;\alpha) \coloneq \text{mult'}(l;\alpha,\alpha) \\
  &\mathbf{def}\ \text{mult'}(l;\alpha, \beta) \coloneq \\
  & \quad \cut{l}{\SCcase{\text{Nil} \Rightarrow \cut{1}{\beta}, \text{Cons}(x,xs) \Rightarrow \ifzero{x}{\cut{0}{\alpha}}{\text{mult'}(xs;\alpha,\tilde\mu z.*(x,z;\beta))}}}
\end{align*}
Here is how you read this snippet:
Besides the list argument $l$, the definition $\mathbf{def}\ \text{mult}(l;\alpha) \coloneq \ldots$ takes an argument $\alpha$ which indicates how the computation should continue once the result of the multiplication is computed (we again use $;$ to separate these two kinds of arguments).
The helper function $\text{mult'}$ takes a list argument $l$ and two arguments $\alpha$ and $\beta$; the argument $\beta$ indicates where the function should return to on a normal recursive call while $\alpha$ indicates the return point of a short-circuiting computation.
In the body of $\text{mult'}$ we use $\cut{l}{\SCcase{ \text{Nil} \Rightarrow \ldots, \text{Cons}(x,xs) \Rightarrow \ldots}}$ to perform a case split on the list $l$.
If the list is Nil, then we use $\cut{1}{\beta}$ to return 1 to $\beta$, which is the return for a normal recursive call.
If the list has the form Cons$(x, xs)$ and $x$ is zero, we return with $\cut{0}{\alpha}$, where $\alpha$ is the return point which short-circuits the computation.
If $x$ isn't zero, then we have to perform the recursive call $\text{mult'}(xs;\alpha,\tilde\mu z.*(x,z;\beta))$, where we use $\tilde\mu z.*(x,z;\beta)$ to bind the result of the recursive call to the variable $z$ before multiplying it with $x$ and returning to $\beta$.
Don't be discouraged if this looks complicated at the moment; the main part of this paper will cover everything in much more detail.

\begin{wrapfigure}{R}{0.55\textwidth}
    \includegraphics[width=0.55\textwidth]{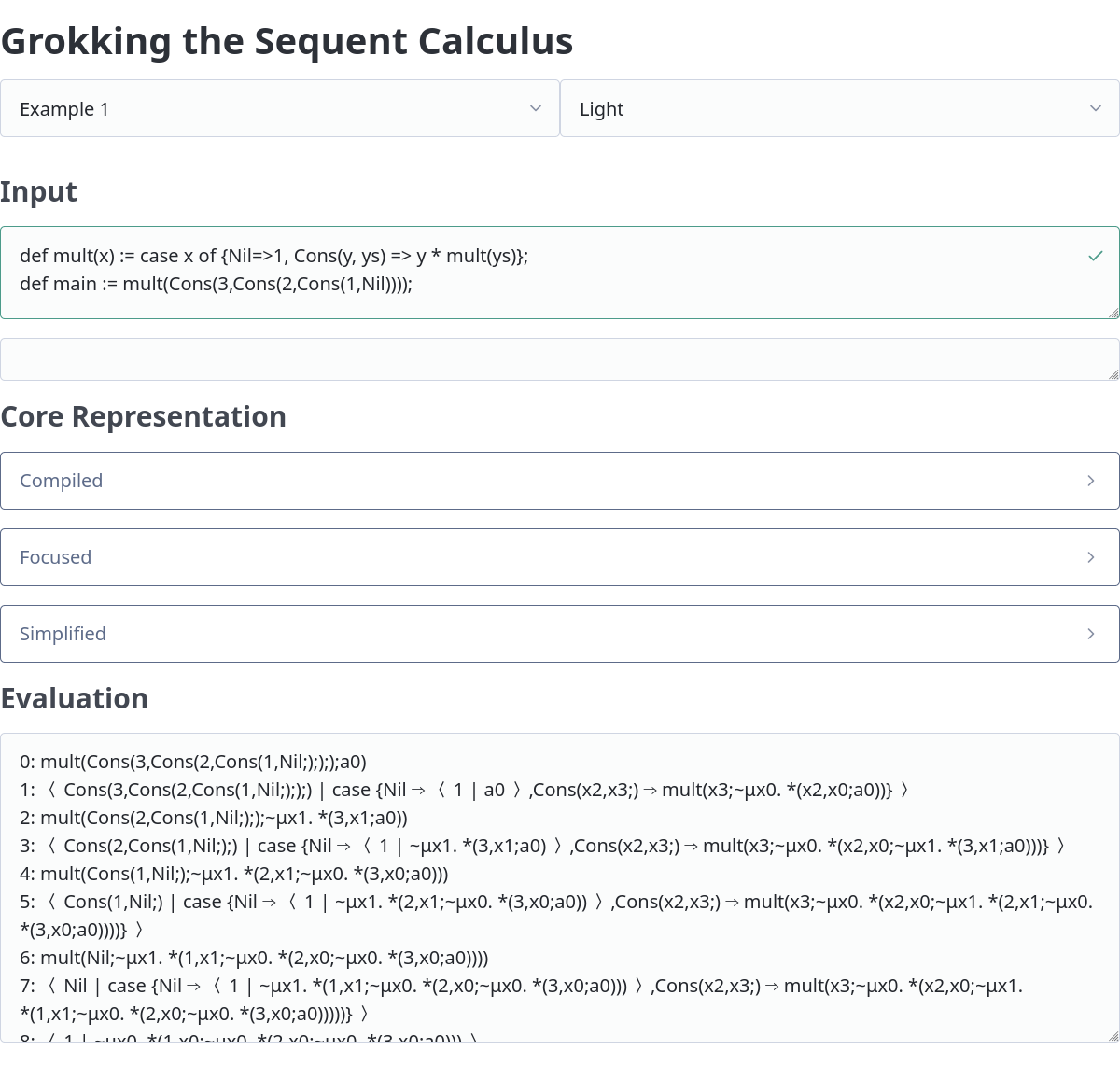}
    \Description{Screenshot of the online evaluator.}
    \caption{Screenshot of the online evaluator.}
    \label{fig:screenshot}
\end{wrapfigure}

The \lambdamumu\ that you have just seen was first introduced by \citet{Curien2000duality} as a solution to a long-standing open question:
What should a term language for the sequent calculus look like?
The sequent calculus is one of two influential proof calculi introduced by \citet{Gentzen1935a,Gentzen1935b} in a single paper, the other calculus being natural deduction.
The term language for natural deduction is the ordinary lambda calculus, but it was difficult to find a good term language for the sequent calculus.
After it had been found, the \lambdamumu\ was proposed as a better foundation for compiler intermediate languages, for example by \citet{Downen2016sequent}.
Despite this, most language designers and compiler writers are still unfamiliar with it.
This is the situation that we hope to remedy with this pearl.

We frequently discuss ideas which involve the \lambdamumu\ with students and colleagues and therefore have to introduce them to its central ideas.
But we usually cannot motivate the \lambdamumu\ as a term assignment system for the sequent calculus, since most of them are not familiar with it.
We instead explain the \lambdamumu\ on the whiteboard by compiling small functional programs into it.
Such an introduction is regrettably still missing in the published literature; most existing presentations either presuppose knowledge of the sequent calculus or otherwise spend a lot of space introducing it first.
We believe that if one can understand the lambda calculus without first learning about natural deduction proofs, then one should also be able to understand the \lambdamumu\ without knowing the sequent calculus\footnote{For the interested reader, we show in \cref{sec:appendix:sequent-calculus-relationship} how the sequent calculus and the \lambdamumu\ are connected.}.

Why are we excited about the \lambdamumu, and why do we think that more people should become familiar with its central ideas and concepts?
The main feature which distinguishes the \lambdamumu\ from the lambda calculus is its first-class treatment of evaluation contexts.
An evaluation context is the remainder of the program which runs after the current subexpression we are focused on finishes evaluating.

This becomes clearer with an example:
When we want to evaluate the expression $(2+3) * 5$, we first have to focus on the subexpression $2+3$ and evaluate it to its result $5$.
The remainder of the program, which will run after we have finished the evaluation, can be represented with the evaluation context $\square * 5$.
We cannot bind an evaluation context like $\square * 5$ to a variable in the lambda calculus, but in the \lambdamumu\ we can bind such evaluation contexts to covariables.
Furthermore, the $\mu$-operator gives direct access to the evaluation context in which the expression is currently evaluated.
Having such direct access to the evaluation context is not always necessary for a programmer who wants to write an application, but it is often important for compiler implementors who write optimizations to make programs run faster.
One solution that compiler writers use to represent evaluation contexts in the lambda calculus is called continuation-passing style.
In continuation-passing style, an evaluation context like $\square * 5$ is represented as a function $\lambda x. x * 5$.
This solution works, but the resulting types which are used to type a program in this style are arguably hard to understand.
Being able to easily inspect these types can be very valuable, especially for intermediate representations, where terms tend to look complex.
The promise of the \lambdamumu\ is to provide the expressive power of programs in continuation-passing style without having to deal with the type-acrobatics that are usually associated with it.

The remainder of this paper is structured as follows:
\begin{itemize}
    \item In \cref{sec:syntax} we introduce the surface language \surfacelang\ and show how we can translate it into the sequent-calculus-based language \targetlang.
          The surface language is mostly an expression-oriented functional programming language, but we have added some features such as codata types and control operators whose translations provide important insights into how the \lambdamumu\ works.
          In this section, we also compare how redexes are evaluated in both languages.

    \item In \cref{sec:focusing} we discuss static and dynamic focusing, which are two closely related techniques for lifting subexpressions which are not values into a position where they can be evaluated.
    \item \Cref{sec:typing} introduces the typing rules for \surfacelang\ and \targetlang\ and proves standard results about typing and evaluation.
    \item We show why we are excited about the \lambdamumu\ in \cref{sec:insights}.
          We present various programming language concepts which become much clearer when we present them in the \lambdamumu:
          We show that let-bindings are precisely dual to control operators, that data and codata types are two perfectly dual ways of specifying types, and that the case-of-case transformation is nothing more than a $\mu$-reduction.
          These insights are not novel for someone familiar with the \lambdamumu, but not yet as widely known as they should be.
    \item Finally, in  \cref{sec:related-work} we discuss related work and provide pointers for further reading.
          We conclude in \cref{sec:conclusion}.
\end{itemize}

This paper is accompanied by a Haskell implementation which we also make available as an interactive website (cf.~\cref{fig:screenshot}).
You can run the examples presented in this paper in the online evaluator.

%% file: sections/syntax.tex
In this section, we introduce \surfacelang, an expression-oriented functional programming language, together with its translation into the sequent-calculus-based intermediate language \targetlang.
We present both languages and the translation function $\translate{-}$ in multiple steps, starting with arithmetic expressions and adding more features in later subsections.
We postpone the typing rules for both languages until \cref{sec:typing}.

\subsection{Arithmetic Expressions}
\label{subsec:syntax-arith-expressions}

We begin with arithmetic expressions which consist of variables, integer literals,  binary operators and $\mathbf{ifz}$, a conditional expression which checks whether its first argument is equal to zero.
The syntax of arithmetic expressions for \surfacelang\ and \targetlang\ is given in \cref{def:syntax:arithmetic-expressions}.
\medskip

\begin{frameddefinition}[Arithmetic Expressions]\hfill
  \begin{equation*}
    x,y,z,\ldots \in \textsc{Variables} \quad \vartop,\alpha,\beta,\gamma,\ldots \in \textsc{Covariables} \quad \odot  \in \lbrace \ast, +, - \rbrace
  \end{equation*}
  \vspace{0.05cm}

  \begin{minipage}[]{0.37\textwidth}
    \[
      \begin{array}{r l}
        \multicolumn{2}{c}{\surfacelang} \\[0.2cm]
        t & \Coloneqq x \mid \natlit{n} \mid t \odot t \mid \ifzero{t}{t}{t} \\
        & \phantom{\Coloneqq} \\
        & \phantom{\Coloneqq} \\
      \end{array}
    \]
  \end{minipage}
  \hfill\vline\hfill
  \begin{minipage}{0.56\textwidth}
    \[
      \begin{array}{r l r}
        \multicolumn{3}{c}{\targetlang} \\[0.2cm]
        p & \Coloneqq x \mid \natlit{n} \mid \mu\alpha.s & \emph{Producer} \\
        c & \Coloneqq \alpha & \emph{Consumer} \\
        s & \Coloneqq \odot(p,p;c) \mid \ifzero{p}{s}{s} \mid \cut{p}{c} & \emph{Statement}
      \end{array}
    \]
  \end{minipage}
  \vspace{0.2cm}

  \begin{align*}
    \translate{x} & \coloneq x & \translate{t_1 \odot t_2} &\coloneq \mu \alpha.\odot(\translate{t_1},\translate{t_2};\alpha) & \fresh{\alpha}\\
    \translate{\natlit{n}} &\coloneq \natlit{n} &
    \translate{\ifzero{t_1}{t_2}{t_3}} &\coloneq \mu\alpha.\ifzero{\translate{t_1}}{\cut{\translate{t_2}}{\alpha}}{\cut{\translate{t_3}}{\alpha}} & \fresh{\alpha}
  \end{align*}
  \label{def:syntax:arithmetic-expressions}
\end{frameddefinition}
\medskip

In \surfacelang\ there is only one syntactic category: terms $t$.
These terms can either be variables $x$, literals $\natlit{n}$, binary operators $t + t$, $t * t$ and $t - t$, or a conditional $\ifzero{t}{t_0}{t_1}$.
This conditional evaluates to $t_0$ if $t$ evaluates to $\natlit{0}$, or to $t_1$ otherwise.
In contrast to this single category, \targetlang\ uses three different syntactic categories: producers $p$, consumers $c$ and statements $s$.
These categories are directly inherited from the \lambdamumu, and it is important to understand their differences:
\begin{description}
  \item[Producers] All constructs in \targetlang\ which \emph{construct} or \emph{produce} an element of some type belong to the syntactic category of producers.
  In other words, producers correspond to \enquote{introduction forms} or \enquote{proof terms}, and every term of the language \surfacelang\ is translated to a producer in \targetlang.
  \item[Consumers] Consumers are probably less intuitive than producers since they do not correspond directly to any term of the language \surfacelang.
  The basic idea is that if some consumer $c$ has type $\tau$, then $c$ \emph{consumes} or \emph{destructs} a producer of type $\tau$.
  If you have encountered evaluation contexts or continuations before, then it is helpful to think of consumers of type $\tau$ as continuations or evaluation contexts for a producer of type $\tau$.
  And if you are familiar with the Curry-Howard correspondence, then you can think of consumers as refutations or direct evidence that a proposition is false.
  \item[Statements] Statements are the ingredient which make computation \emph{happen}; without statements, we would only have static objects without any dynamic behavior.
  Here is a non-exhaustive list of examples for statements:
  Every IO action which reads from or prints to the console or a file should be represented as a statement in \targetlang.
  Computations on primitive types such as machine integers should be statements.
  Finally, everything which is a redex in an expression-based language should also correspond to a statement in \targetlang.
  Since statements themselves only compute and do not return anything they do not have a type.
\end{description}

After these general remarks, let us now look at how arithmetic expressions are represented in the language \targetlang.
Variables $x$ and literals $\natlit{n}$ both belong to the category of producers, but binary operators are represented as statements $\odot(p_1,p_2;c)$.
First, let us explain why they are represented as statements instead of producers.
The idea is that a binary operator on primitive integers has to be evaluated directly by the arithmetic logic unit (ALU) of the underlying machine.
And any operation which directly invokes the machine should belong to the same syntactic category as a print or other IO instruction: statements.
The machine does not return a result; rather, it reads inputs from registers and makes the result available in a register for further computation.
This is also reflected in the second surprising aspect: the operator has three instead of two arguments.
The two producers $p_1$ and $p_2$ correspond to the usual arguments, but the third consumer argument $c$ says what should happen to the result once the binary operator has been evaluated.
This is similar to the continuation argument of a function in continuation-passing style.
Binary operators $\odot(p_1,p_2;c)$ also display a syntactic convention we use: whenever some construct has arguments of different syntactic categories, we use a semicolon instead of a comma to separate them.

We can immediately see that the result of $\translate{p_1 + p_2}$ should contain the statement $+(\translate{p_1},\translate{p_2};?)$, but we still have to figure out which consumer to plug in at the third-argument place, and how to convert this statement into a producer.
We can do this with a $\mu$-abstraction in \targetlang, which turns a statement into a producer while binding a covariable $\alpha$: $\mu \alpha.+(\translate{p_1},\translate{p_2};\alpha)$.

The statement $\mathbf{ifz}$ works similarly to binary operators:
It is a computation which checks if the producer $p$ is zero and then continues with one of its two branches.
These branches are also statements, indicating which computation to run after the condition has been evaluated.
In the language \surfacelang\ the two branches were terms, so we now have to find a way to transform two producers into two statements.
We can do this by using a cut $\cut{p}{c}$ which combines a producer and a consumer of the same type to obtain a statement in each branch: $\ifzero{\translate{t_1}}{\cut{\translate{t_2}}{\,?\,}}{\cut{\translate{t_3}}{\,?\,}}$.
We can then use the same covariable $\alpha$ in both statements to represent the fact that the we want the result in either branch to return to the same point in the program; we use a surrounding $\mu$-binding again to bind this covariable: $\mu \alpha.\ifzero{\translate{t_1}}{\cut{\translate{t_2}}{\alpha}}{\cut{\translate{t_3}}{\alpha}}$.

Let us now see how arithmetic expressions are evaluated.
\cref{def:syntax:arithmetic-expressions-eval} introduces the syntax of values and covalues, and shows how to reduce immediate redexes.
We use a simple syntactic convention here:
The metavariable for a value of terms $t$ is $\valueof{t}$, the values of producers $p$ are written $\valueof{p}$ and the covalues which correspond to consumers $c$ are written $\valueof{c}$.
We use the symbol $\reducesto$ for reduction in both \surfacelang\ and \targetlang\ (and write $\reducesto^{\ast}$ when multiple steps are performed at once).

\medskip
\begin{frameddefinition}[Evaluation for Arithmetic Expressions]\hfill

  \begin{minipage}{0.45\textwidth}
    \[
      \begin{array}{l c l r}
      \multicolumn{4}{c}{\surfacelang} \\[0.2cm]
      \valueof{t} & \Coloneqq & \natlit{n} &  \emph{Values}\\
      & \phantom{X} & & \\
      \end{array}
    \]
    \begin{align*}
      \ifzero{\natlit{0}}{t_1}{t_2} &\reducesto t_1 \\
      \ifzero{\natlit{n}}{t_1}{t_2} &\reducesto t_2\quad  \text{(if $n\neq 0$)}\\
      \natlit{n}\odot \natlit{m} &\reducesto \natlit{n\odot m} \\
      \phantom{X} &
    \end{align*}
  \end{minipage}
  \hfill\vline\hfill
  \begin{minipage}{0.45\textwidth}
    \[
      \begin{array}{l c l r}
        \multicolumn{4}{c}{\targetlang} \\[0.2cm]
        \valueof{p} & \Coloneqq & \natlit{n} & \emph{Values}\\
        \valueof{c} & \Coloneqq & \alpha & \emph{Covalues}
      \end{array}
    \]
    \begin{align*}
      \ifzero{\natlit{0}}{s_1}{s_2} &\reducesto s_1 \\
      \ifzero{\natlit{n}}{s_1}{s_2} &\reducesto s_2 \quad \text{(if $n\neq 0$)}\\
      \odot(\natlit{n},\natlit{m};c) &\reducesto \cut{\natlit{n \odot m}}{c} \\
       \cut{\mu\alpha.s}{\valueof{c}} &\reducesto s[\valueof{c}/\alpha]
    \end{align*}
  \end{minipage}
  \label{def:syntax:arithmetic-expressions-eval}
\end{frameddefinition}
\medskip

Values and the evaluation of redexes in \surfacelang\ is straightforward, the only noteworthy aspect is that the two rules for $\ifzero{\cdot}{t_1}{t_2}$ do not require $t_1$ and $t_2$ to be values.
Thus, let us proceed with the discussion of the language \targetlang.

The first interesting aspect of the language \targetlang\ is that there are both values and covalues.
This can be explained by the role that values play in operational semantics: they specify the subset of terms that we are allowed to substitute for a variable.
And since we have both variables which stand for producers and covariables which stand for consumers, we need both values and covalues as the respective subsets which we are allowed to substitute for a variable or covariable.

The second interesting aspect of the language \targetlang\ is that only statements are reduced, not producers or consumers.
This substantiates our remark from above that it is statements that introduce dynamism into the language by driving computation.
It also contributes to the feeling that reduction in the language is close to the evaluation of an abstract machine and that the statements of \targetlang\ correspond to the states of such an abstract machine.

We are still faced with a small problem when we want to show that a term of \surfacelang\ evaluates to the same result as its translation into \targetlang: We have only specified the reduction for statements but not for producers.
We can easily solve this problem by introducing a special covariable $\vartop$ which acts as the \enquote{top-level} consumer of an evaluation.
Using $\vartop$ we can then evaluate the statement $\cut{\translate{t}}{\vartop}$ instead of the producer $\translate{t}$.

\begin{example}
  \label{ex:arithmetic-expression}
  Consider the two terms $\natlit{2}*\natlit{3}$ and $\ifzero{\natlit{2}}{\natlit{5}}{\natlit{10}}$ of \surfacelang.
  Their respective translations into \targetlang\ are $\mu\alpha.*(\natlit{2},\natlit{3};\alpha)$ and
  $\mu \alpha. \ifzero{\natlit{2}}{\cut{\natlit{5}}{\alpha}}{\cut{\natlit{10}}{\alpha}}$.
  When we wrap them into a statement using the top-level continuation $\vartop$, we observe the following evaluation:
  \begin{align*}
    &\cut{\mu\alpha.*(\natlit{2},\natlit{3};\alpha)}{\vartop}\ \reducesto\ *(\natlit{2},\natlit{3};\vartop) \ \reducesto\ \cut{\natlit{6}}{\vartop} \\
    &\cut{\mu \alpha. \ifzero{\natlit{2}}{\cut{\natlit{5}}{\alpha}}{\cut{\natlit{10}}{\alpha}}}{\vartop}
    \ \reducesto\
    \ifzero{\natlit{2}}{\cut{\natlit{5}}{\vartop}}{\cut{\natlit{10}}{\vartop}}
    \ \reducesto\ \cut{\natlit{10}}{\vartop}
  \end{align*}
  We have successfully evaluated the first term to the result $\natlit{6}$ and the second term to the result $\natlit{10}$.
\end{example}

In the following, we will often leave out the first reduction step in examples, thus silently replacing the covariable bound by the outermost $\mu$-binding with the top-level consumer $\vartop$.

Here is a bigger problem that we haven't addressed yet.
The evaluation rules in the present section do not allow to evaluate nested expressions like $(\natlit{2} * \natlit{4}) + \natlit{5}$ in \surfacelang\ or its translation $\mu \alpha.+(\mu \beta.*(\natlit{2},\natlit{4};\beta),\natlit{5};\alpha)$ in \targetlang.
We will discuss this problem and its solution in more detail in \cref{sec:focusing}.

\subsection{Let Bindings}
\label{subsec:syntax-let-bindings}

Let-bindings are important since we can use them to eliminate duplication and make code more readable.
In this section we introduce let-bindings to \surfacelang\ for an additional reason: they allow us to introduce the second construct which gives the \lambdamumu\ its name: $\tilde\mu$-abstractions.
\medskip

\begin{frameddefinition}[Let-Bindings and $\tilde\mu$-abstractions]
  \hfill
  \smallskip

  \begin{minipage}{0.45\textwidth}
    \[
      \begin{array}{l c l}
        \multicolumn{3}{c}{\surfacelang} \\[0.2cm]
        t & \Coloneqq & \ldots \mid \letin{x}{t}{t} \\
        \phantom{\Coloneqq}
      \end{array}
    \]
    \begin{align*}
      \letin{x}{\valueof{t}}{t} \reducesto t[\valueof{t}/x]
    \end{align*}
  \end{minipage}
  \hfill\vline\hfill
  \begin{minipage}{0.45\textwidth}
    \[
      \begin{array}{l c l}
        \multicolumn{3}{c}{\targetlang} \\[0.2cm]
        c & \Coloneqq & \ldots \mid \tilde{\mu}x.s \\
        \valueof{c} & \Coloneqq & \ldots \mid \tilde\mu x.s\\
      \end{array}
    \]
    \begin{align*}
      \cut{\valueof{p}}{\tilde{\mu}x.s}\reducesto s[\valueof{p}/x]
    \end{align*}
  \end{minipage}

  \begin{align*}
    \translate{\letin{x}{t_1}{t_2}} \coloneq \mu \alpha.\cut{\translate{t_1}}{\tilde\mu x.\cut{\translate{t_2}}{\alpha}} &\quad \fresh{\alpha}
  \end{align*}
\end{frameddefinition}
\medskip

The let-bindings in \surfacelang\ are standard and are evaluated by substituting the \emph{value} $\valueof{t}$ for the variable $x$ in the body which is a term.
The analogue of a let-binding in \surfacelang\ is a $\tilde\mu$-binding in \targetlang\ which also binds a variable, with the difference that the body of a $\tilde\mu$-binding is a statement.
It can easily be seen that $\tilde\mu$-bindings are the precise dual of $\mu$-bindings that we have already introduced.

With both $\mu$- and $\tilde\mu$-bindings in \targetlang\ we have to face a potential problem, namely statements of the form $\cut{\mu\alpha.s_1}{\tilde\mu x.s_2}$.
Such a statement is called a \emph{critical pair} since it can potentially be reduced to both $s_1[\tilde\mu x.s_2/\alpha]$ and $s_2[\mu \alpha.s_1/x]$ which can be a source of non-confluence.
A closer inspection of the rules shows that we avoid this pitfall and always evaluate the statement to $s_1[\tilde\mu x.s_2/\alpha]$.
We do not allow to reduce the statement to $s_2[\mu \alpha.s_1/x]$ since only values $\valueof{p}$ can be substituted for variables, and $\mu \alpha.s_1$ is not a value.
This restriction precisely mirrors the restriction on the evaluation of let-bindings in \surfacelang.
In other words, we use call-by-value evaluation order.
We will address the critical pair and how it relates to different evaluation orders again in \cref{subsec:insights:strict-vs-lazy}.

\begin{example}
  \label{ex:let-syntax}
  Consider the term $\letin{x}{\natlit{2}*\natlit{2}}{x*x}$ whose translation into \targetlang\ is the producer $\mu\alpha.\underline{\cut{\mu\beta.*(\natlit{2},\natlit{2};\beta)}{\tilde\mu x.\cut{\mu\gamma.*(x,x;\gamma)}{\alpha}}}$.
  This producer contains a critical pair which we have underlined.
  Because we are using call-by-value, we can observe how the following reduction steps resolve the critical pair by evaluating the  $\mu$-abstraction first.
  \begin{align*}
    \cut{\mu\beta.*(\natlit{2},\natlit{2};\beta)}{\tilde\mu x.\cut{\mu\gamma.*(x,x;\gamma)}{\vartop}}
    \reducesto *(\natlit{2},\natlit{2};\tilde\mu x.\cut{\mu\gamma.*(x,x;\gamma)}{\vartop})\reducesto\\
    \cut{\natlit{4}}{\tilde\mu x. \cut{\mu\gamma.*(x,x;\gamma)}{\vartop}}
    \reducesto \cut{\mu\gamma.*(\natlit{4},\natlit{4};\gamma)}{\vartop}
    \reducesto *(\natlit{4},\natlit{4};\vartop)
    \reducesto \cut{\natlit{16}}{\vartop}
  \end{align*}
  We can observe that the arithmetic expression $2 * 2$ has been evaluated only once, which is precisely what we expect from call-by-value.
\end{example}

\subsection{Top-level Definitions}
\label{subsec:syntax-recursion}

We introduce recursive top-level definitions to \surfacelang\ and \targetlang\ for two reasons.
They allow us to write more interesting examples and they illustrate a difference in how recursive calls are handled.
The extension is specified in \cref{def:syntax:toplevel}.
\medskip
\begin{frameddefinition}[Top-level Definitions]
  \label{def:syntax:toplevel}
  We assume for both languages that $f,g,h,\ldots \in \textsc{Names}$.\hfill

  \begin{minipage}{0.45\textwidth}
    \[
      \begin{array}{l c l r}
        \multicolumn{4}{c}{\surfacelang} \\[0.2cm]
        F & \Coloneqq & \mathbf{def}\ f(\overline{x};\overline{\alpha}) \coloneq t & \emph{Definitions}\\
        P & \Coloneqq & \emptyset\ \mid\ F,P                               & \emph{Programs}\\
        t & \Coloneqq & \ldots \mid f(\overline{t};\overline{\alpha})      & \emph{Terms}\\
      \end{array}
    \]
    \begin{equation*}
      f(\overline{\valueof{t}};\overline{\alpha})\reducesto t[\overline{\valueof{t}}/\overline{x},\overline{\alpha}/\overline{\beta}] \quad (\text{if } f(\overline{x};\overline{\beta}) \coloneq t \in P)
    \end{equation*}
  \end{minipage}
  \hfill\vline\hfill
  \begin{minipage}{0.45\textwidth}
    \[
      \begin{array}{l c l r}
        \multicolumn{4}{c}{\targetlang} \\[0.2cm]
        F & \Coloneqq & \mathbf{def}\ f(\overline{x};\overline{\alpha}) \coloneq s & \emph{Definitions}\\
        P & \Coloneqq & \emptyset\ \mid\ F,P & \emph{Programs}\\
        s & \Coloneqq & \ldots\ \mid\ f(\overline{p};\overline{c}) & \emph{Statements}\\
      \end{array}
    \]
    \begin{equation*}
      f(\overline{\valueof{p}};\overline{\valueof{c}}) \reducesto s[\overline{\valueof{p}}/\overline{x},\overline{\valueof{c}}/\overline{\alpha}] \quad (\text{if } f(\overline{x};\overline{\alpha}) \coloneq s\in P)
    \end{equation*}
  \end{minipage}
  \medskip

  \[
    \begin{array}{ccc}
      \translate{\mathbf{def}\ f(\overline{x};\overline{\alpha}) \coloneq t} \coloneq \mathbf{def}\ f(\overline{x};\overline{\alpha},\alpha) \coloneq \cut{\translate{t}}{\alpha} & \fresh{\alpha}  \\
      \translate{f(\overline{t};\overline{\alpha})} \coloneq  \mu\alpha.f(\overline{\translate{t}};\overline{\alpha},\alpha) & \fresh{\alpha}
    \end{array}
  \]

\end{frameddefinition}
\medskip

Top-level definitions should not be confused with first-class functions which will be introduced later, since they cannot be passed as an argument or returned as a result.
They are a part of a program that consists of a list of such top-level definitions.
The top-level definitions in \surfacelang\ curiously also take covariables as arguments even though the language does not contain consumers; you can ignore that for now.
If you remember the example from the introduction, then you might recall that we use them for passing labels, but we will only formally introduce that construct in \cref{subsec:syntax-exception-handling}.

We evaluate the call of a top-level definition by looking up the body in the program and substituting the arguments of the call for the parameters in the body of the definition.
The body of a top-level definition is a term in \surfacelang\ and a statement in \targetlang.
This difference explains why we have to add an additional parameter $\alpha$ to every top-level definition when we translate it; this parameter $\alpha$ also corresponds to the additional continuation argument when we ordinarily translate a function into continuation-passing style.
We could also have specified that the body of a top-level definition in \targetlang\ should be a producer.
We don't do that because when we eventually translate \targetlang\ to machine code we want every top-level definition to become the target of a jump with arguments \emph{without building up a function call stack}.
The following example shows how this works:

\begin{example}
  \label{ex:recursive-syntax}
  Using a top-level definition, we can represent the factorial function in \targetlang.
  \begin{align*}
    \mathbf{def}\ \text{fac}(n;\alpha) \coloneq \ifzero{n}{\cut{\natlit{1}}{\alpha}}{-(n,\natlit{1};\tilde\mu x.\text{fac}(x;\tilde\mu r.*(n,r;\alpha)))}
  \end{align*}
  For the argument $\natlit{1}$ this evaluates in the following way:
  \begin{align*}
    \text{fac}(\natlit{1},\vartop)
    &\reducesto \ifzero{\natlit{1}}{\cut{\natlit{1}}{\vartop}}{-(\natlit{1},\natlit{1};\tilde\mu x.\text{fac}(x;\tilde\mu r.*(\natlit{1},r;\vartop)))} \\
    &\reducesto -(\natlit{1},\natlit{1};\tilde\mu x.\text{fac}(x;\tilde\mu r.*(\natlit{1},r;\vartop))) \\
    &\reducesto \cut{\natlit{0}}{\tilde\mu x.\text{fac}(x;\tilde\mu r.*(\natlit{1},r;\vartop))} \\
    \tag{$\ast$}
    &\reducesto \text{fac}(\natlit{0};\tilde\mu r.*(\natlit{1},r;\vartop)) \\
    &\reducesto \ifzero{\natlit{0}}{\cut{\natlit{1}}{\tilde\mu r.*(\natlit{1},r;\vartop)}}{\ldots} \\
    &\reducesto \cut{\natlit{1}}{\tilde\mu r.*(\natlit{1},r;\vartop)} \\
    &\reducesto *(\natlit{1},\natlit{1};\vartop) \reducesto \cut{\natlit{1}}{\vartop}
  \end{align*}
  At the point ($\ast$) of the evaluation we can now see how the recursive call is evaluated.
  In \surfacelang\ this recursive call would have the form $1 *\text{fac}(0)$ and require a function stack, but in \targetlang\ we can jump to the definition of fac with the consumer $\tilde\mu r.*(\natlit{1},r;\vartop)$ as an additional argument which contains the information that the result of the recursive call should be bound to the variable $r$ and then multiplied with $\natlit{1}$.
  Note again that this consumer argument corresponds to a continuation in continuation-passing style (in that sense it might be viewed as a reified stack) and so the basic techniques used in CPS-based intermediate representations and compilers can be applied for its implementation.
\end{example}

\subsection{Algebraic Data and Codata Types}
\label{subsec:syntax-data-codata}

We now extend \surfacelang\ and \targetlang\ with two new features: algebraic data and codata types.
Algebraic data types are familiar from most typed functional programming languages.
Algebraic codata types \cite{Hagino1989codatatypes} are a little more unusual; they are defined by a set of observations or methods called destructors and are quite similar to interfaces in object-oriented programming \cite{Cook2009understanding}.
We introduce them both in the same section because they help to illustrate some of the deep theoretical dualities and symmetries of the sequent calculus and the \lambdamumu.

To get acquainted with our syntax, let us first briefly look at two short examples in \surfacelang.
The following definition calculates the sum over a $\mathtt{List}$ it receives as input.
\begin{align*}
        &\mathbf{def}\ \text{sum}(x)\coloneq \case{x}{\mathtt{Nil} \Rightarrow \natlit{0}, \mathtt{Cons}(y,ys) \Rightarrow y + \text{sum}(ys)}
\end{align*}
It does so by pattern matching using the $\case{...}{...}$ construct which is entirely standard.
As an example of codata types, consider this definition:
\begin{align*}
        &\mathbf{def}\ \text{repeat}(x) \coloneq \cocase{\mathtt{hd} \Rightarrow x, \mathtt{tl} \Rightarrow \text{repeat}(x)}
\end{align*}
It constructs an infinite $\mathtt{Stream}$ whose elements are all the same as the input $x$ of the function.
A $\mathtt{Stream}$ is defined by two destructors, $\mathtt{hd}$ yields the head of the stream and $\mathtt{tl}$ yields the remaining stream without the head.
The stream is constructed by copattern matching \cite{Abel2013copatterns} using the $\cocase{...}$ construct.

\medskip
\begin{frameddefinition}[Algebraic Data and Codata Types]
  \label{def:syntax:data-codata}\hfill\\
  \hspace*{-0.4cm}
  \begin{minipage}{0.45\textwidth}
    \[
      \begin{array}{rl}
        \multicolumn{2}{c}{\surfacelang} \\[0.2cm]
        t & \Coloneqq \ldots \mid K(\overline{t}) \mid \case{t}{\overline{K(\overline{x}) \Rightarrow t}} \\
        & \mid t.D(\overline{t}) \mid \cocase{\overline{D(\overline{x}) \Rightarrow t}} \\
        \valueof{t} & \Coloneqq \ldots \mid K(\overline{\valueof{t}}) \mid \cocase{\overline{D(\overline{x}) \Rightarrow t}}\\
        \phantom{\Coloneqq}
      \end{array}
    \]
    \begin{align*}
        \case{K(\overline{\valueof{t}})}{K(\overline{x})\Rightarrow t,\ldots} \reducesto t[\overline{\valueof{t}}/\overline{x}] \\
        \cocase{D(\overline{x}) \Rightarrow t,\ldots}.D(\overline{\valueof{t}}) \reducesto t[\overline{\valueof{t}}/\overline{x}]
    \end{align*}
  \end{minipage}
  \hspace*{0.05cm}
  \hfill\vline\hfill
  \begin{minipage}{0.53\textwidth}
    \[
      \begin{array}{rl}
        \multicolumn{2}{c}{\targetlang} \\[0.2cm]
          p &\Coloneqq \ldots \mid K(\overline{p};\overline{c}) \mid \SCcocase{\overline{D(\overline{x};\overline{\alpha})\Rightarrow s}} \\
          c &\Coloneqq \ldots \mid D(\overline{p};\overline{c}) \mid \SCcase{\overline{K(\overline{x};\overline{\alpha})\Rightarrow s}} \\
          \valueof{p} &\Coloneqq \ldots \mid K(\overline{\valueof{p}};\overline{c}) \mid \SCcocase{\overline{D(\overline{x};\overline{\alpha})\Rightarrow s}}\\
          \valueof{c} & \Coloneqq \ldots \mid D(\overline{p};\overline{c}) \mid \SCcase{\overline{K(\overline{x};\overline{\alpha})\Rightarrow s}}
      \end{array}
    \]
    \begin{align*}
      \cut{K(\overline{\valueof{p}};\overline{\valueof{c}})}{\SCcase{K(\overline{x};\overline{\alpha})\Rightarrow s,\ldots}} \reducesto s[\overline{\valueof{p}}/\overline{x};\overline{\valueof{c}}/\overline{\alpha}] \\
      \cut{\SCcocase{D(\overline{x};\overline{\alpha}) \Rightarrow s,\ldots}}{D(\overline{\valueof{p}};\overline{\valueof{c}})} \reducesto s[\overline{\valueof{p}}/\overline{x};\overline{\valueof{c}}/\overline{\alpha}]
    \end{align*}
  \end{minipage}

  \begin{align*}
    \translate{K(t_1,\ldots,t_n)} & \coloneq K(\translate{t_1},\ldots,\translate{t_n}) \\
    \translate{\case{t}{ \overline{K_i(\overline{x_{i,j}}) \Rightarrow t_i} }} & \coloneq \mu \alpha.\cut{\translate{t}}{\SCcase{\overline{K_i(\overline{x_{i,j}}) \Rightarrow \cut{\translate{t_i}}{\alpha}}}} & \fresh{\alpha}\\
    \translate{t.D(t_1,\ldots,t_n)} &\coloneq \mu \alpha. \cut{\translate{t}}{D(\translate{t_1},\ldots,\translate{t_n};\alpha)} & \fresh{\alpha}\\
    \translate{\cocase{\overline{D_i(\overline{x_{i,j}}) \Rightarrow t_i}}} &\coloneq \SCcocase{\overline{D_i(\overline{x_{i,j}};\alpha_i) \Rightarrow \cut{\translate{t_i}}{\alpha_i}}} & \fresh{\overline{\alpha_i}}
  \end{align*}
\end{frameddefinition}
\medskip

The general syntax is given in \cref{def:syntax:data-codata}.
We assume fixed sets of constructors $K$ containing at least $\mathtt{Nil}$, $\mathtt{Cons}$ and $\mathtt{Tup}$ and destructors $D$ containing at least $\mathtt{hd}$, $\mathtt{tl}$, $\mathtt{fst}$ and $\mathtt{snd}$.
In \surfacelang\ we use constructors $K$ to define both terms $K(\overline{t})$ and case expressions $\case{t}{\overline{K(\overline{x})\Rightarrow t}}$.
Destructors $D$ of codata types are used in destructor terms $t.D(\overline{t})$ and cocase expressions $\cocase{\overline{D(\overline{x})\Rightarrow t}}$.
The term $t$ in $\case{t}{\overline{K(\overline{x})\Rightarrow t}}$ and $t.D(\overline{t})$ is called the \emph{scrutinee} in both cases.

\subsubsection{Data Types}
Let us consider another example to better understand the general syntax:
\begin{align*}
        &\mathbf{def}\ \text{swap}(x)\coloneq \case{x}{\mathtt{Tup}(y,z) \Rightarrow \mathtt{Tup}(z,y)}
\end{align*}
The function $\text{swap}$ takes a $\mathtt{Pair}$ and swaps its elements.
To do so, it pattern matches on its input using the $\case{t}{\overline{K(\overline{x}) \Rightarrow t}}$ construct, and constructs a tuple using a constructor $K(\overline{t})$, where $K$ is specialized to $\mathtt{Tup}$.
Our syntax is quite general, so it is easy to extend it with new constructors; any such extension only requires that we also add corresponding typing rules (\cref{sec:typing}).

In \targetlang, algebraic data types are mostly handled in the same way as in \surfacelang.
The main difference is that the scrutinee is no longer a part of a case expression.
Instead, the case expression is a consumer and the scrutinee is a producer, which are then combined in a statement.
This is exactly what is done in the translation.
When a case and a constructor meet, there is an opportunity for computation, \emph{consuming} the constructed term and continuing with the corresponding right-hand side of the case expression.
This also explains our terminology of \emph{producers} and \emph{consumers}.
Constructors create, or in other words, \emph{produce} data structures while cases destroy, or \emph{consume} them.

There is another difference, however.
Constructors in \targetlang\ can now also take consumers as arguments which is not the case in \surfacelang.
An example of this is the negation type of a type $\tau$ which can be formulated as a data type with one constructor taking a consumer of type $\tau$ as an argument.
A program making use of this type can be found in section 7.2 of \citet{icfp2022}.

\surfacelang\ is a call-by-value language which manifests itself in that a value of an algebraic data type consists of a constructor applied to other values.
A case expression $\case{t}{\ldots}$ can only be evaluated if the scrutinee $t$ is a value, so this means that it must be a constructor whose arguments are all values in the evaluation rule.

Evaluation in \targetlang\ is done the same way, only with the scrutinee term changed to be the producer of a cut.
Note that all consumers in \targetlang\ are covalues (which is why the arguments of destructors in the definition of covalues are not in Fraktur font), so in order for a constructor term to be a value, only its producer arguments need to be values.
This also means that the requirement for the consumer arguments of the constructor to be covalues is vacuously satisfied in the evaluation rule in \targetlang.

\begin{example}
  \label{ex:data-syntax}
  The translation of $\text{swap}$ (including a simplification) is given by
  \begin{align*}
    \mathbf{def}\ \text{swap}(x;\alpha) \coloneq \cut{x}{\SCcase{\mathtt{Tup}(y,z)\Rightarrow \cut{\mathtt{Tup}(z,y)}{\alpha}}}
  \end{align*}
  Evaluating with an argument $\mathtt{Tup}(\natlit{2},\natlit{3})$ and $\vartop$ then proceeds as we would expect
  \begin{align*}
    \cut{\mathtt{Tup}(\natlit{2},\natlit{3})}{\SCcase{\mathtt{Tup}(y,z)\Rightarrow \cut{\mathtt{Tup}(z,y)}{\vartop}}}
    \reducesto \cut{\mathtt{Tup}(\natlit{3},\natlit{2})}{\vartop}
  \end{align*}
\end{example}

\subsubsection{Codata Types}
To illustrate the syntax for codata types further, consider the definition
\begin{align*}
        &\mathbf{def}\ \text{swap_lazy}(x) \coloneq \cocase{\mathtt{fst} \Rightarrow x.\mathtt{snd}, \mathtt{snd} \Rightarrow x.\mathtt{fst} }
\end{align*}
$\text{swap_lazy}$ takes a lazy pair ($\mathtt{LPair}$), which is defined by its projections $\mathtt{fst}$ and $\mathtt{snd}$, and swaps its elements.
It does so with a copattern match $\cocase{\overline{D_i(\overline{x})\Rightarrow t_i}}$ which invokes the opposite destructor on the original pair in each branch.
With a destructor invocation $t.D(\overline{t})$, where $D$ is specialized to $\mathtt{fst}$ or $\mathtt{snd}$, we can then obtain the corresponding component of the new pair.

For codata in \surfacelang, the scrutinee is located in the destructor term instead of the cocase, inverse to data types.
So now destructors are the consumers and cocases are the producers.
This is mirrored in the translation which again separates the scrutinee, since in \targetlang\ codata types and copattern matching are perfectly dual to data types and pattern matching.

All the destructors we have used here do not have producer parameters, but this is just due to the selection of examples.
In the next section, we will see an example of a destructor with a producer parameter.
Moreover, during the translation each destructor is endowed with an additional consumer parameter which again determines how execution continues after the destructor was invoked (and is thus bound by a surrounding $\mu$).
For constructors this is not necessary, as we can use the same consumer variable directly in each branch of a $\mathbf{case}$ (similar to $\mathbf{ifz}$\footnote{We could have modelled $\mathbf{ifz}$ as a $\mathbf{case
}$, too, by modeling numbers as a data type. But since $\mathbf{ifz}$ corresponds to a machine instruction quite directly, it is natural to make it a statement, as explained in \cref{subsec:syntax-arith-expressions}.}) because the scrutinee and the $\mathbf{case}$ are in the same expression.
Destructors (and also constructors) in \targetlang\ can even have more than one consumer parameter.
An example of this is given in \cref{subsec:related-work:linear-logic}.
%

Evaluation is done analogous to data types, with the roles of cases and constructors reversed for cocases and destructors.
Note, however, that for evaluation in \targetlang\ the producer arguments of the destructor also have to be values, so it is not sufficient for the destructor to be a covalue (which it always is).
We will come back to this subtlety in \cref{sec:focusing}.

\begin{example}
  \label{ex:codata-syntax}
  Translating $\text{swap_lazy}$ is done analogously to $\text{swap}$.
  \begin{align*}
    \mathbf{def}\ \text{swap_lazy}(x;\alpha) \coloneq \cut{\SCcocase{\mathtt{fst}(\beta)\Rightarrow \cut{x}{\mathtt{snd}(\beta)}, \mathtt{snd}(\beta)\Rightarrow \cut{x}{\mathtt{fst}(\beta)}}}{\alpha}
  \end{align*}
  Now take $p = \SCcocase{\mathtt{fst}(\alpha) \Rightarrow \cut{\natlit{1}}{\alpha}, \mathtt{snd}(\alpha)\Rightarrow *(\natlit{2},\natlit{3};\alpha)}$ and evaluate $\text{swap_lazy}$ with $\mathtt{snd}$ to retrieve its first element:
  \begin{align*}
    \text{swap_lazy}(p;\mathtt{snd}(\vartop))\ \reducesto \ &
    \cut{\SCcocase{\mathtt{fst}(\beta)\Rightarrow \cut{p}{\mathtt{snd}(\beta)}, \mathtt{snd}(\beta)\Rightarrow \cut{p}{\mathtt{fst}(\beta)}}}{\mathtt{snd}(\vartop)} \\
    \reducesto \ \cut{p}{\mathtt{fst}(\vartop)}\ \reducesto \ & \cut{\natlit{1}}{\vartop}
  \end{align*}
  Because $\mathbf{cocase}$s are values regardless of their right-hand sides (in contrast to constructors), we can apply the destructor $\mathtt{snd}$ without first evaluating the product $*(\natlit{2},\natlit{3};\alpha)$.
  For pairs, we could not do this, as $\mathtt{Tup}(\natlit{1},*(\natlit{2},\natlit{3};\alpha))$ is not a value, so its arguments have to be evaluated first.
  This is why this codata type is called \emph{lazy pair}, as it allows to not evaluate its contents in contrast to regular pairs.
  \end{example}

This section showed an important property of \targetlang\ which does not hold for \surfacelang.
The data and codata types of \targetlang\ are completely symmetric: the syntax for cases is the same as the syntax for cocases and the same is true for constructors and destructors.
The reason for this deep symmetry is the same reason that makes the sequent calculus more symmetric than natural deduction, but in \cref{def:syntax:data-codata} we can observe it in a programming language.

\subsection{First-Class Functions}
\label{subsec:syntax-functions}
A core feature that we have omitted until now are first-class functions which are characterized by lambda abstractions $\lambda x.t$ and function applications $t_1\ t_2$.
But first-class functions do not add any expressive power to a language with codata types, since codata types are a more general concept which subsumes functions as a special case.
We could therefore implement lambda abstractions and function applications as syntactic sugar in both \surfacelang\ and \targetlang.
This is incidentally also what the developers of Java did when they introduced lambdas to the language \cite{Goetz2014jsr}.
We introduce lambda abstractions and function application to the syntax of \surfacelang\ and desugar them to cocases and destructors of a codata type with an $\mathtt{ap}$ destructor during the translation to \targetlang.
\medskip

\begin{frameddefinition}[First-Class Functions]
  \quad \\
  \begin{minipage}{0.45\textwidth}
    \[
      \begin{array}{lcl}
        \multicolumn{3}{c}{\surfacelang} \\[0.2cm]
        t & \Coloneqq & \ldots \mid \lambda x.t \mid t\ t  \\
        \valueof{t} & \Coloneqq & \ldots \mid \lambda x.t \\[0.2cm]
        \multicolumn{3}{c}{(\lambda x.t)\ \valueof{t} \reducesto t[\valueof{t}/x]}
      \end{array}
    \]
  \end{minipage}
  \hfill\vline\hfill
  \begin{minipage}{0.45\textwidth}
    \[
      \begin{array}{c}
        \targetlang \\[0.2cm]
        D \in \{ \ldots, \mathtt{ap} \}
        \phantom{v}\\ [0.2cm]
        \phantom{\lambda}\\
      \end{array}
    \]
  \end{minipage}

  \[
    \begin{array}{cc}
      \translate{\lambda x.t} \coloneq \SCcocase{ \mathtt{ap}(x;\alpha) \Rightarrow \cut{\translate{t}}{\alpha} } & \fresh{\alpha} \\
      \translate{t_1\ t_2} \coloneq \mu \alpha. \cut{\translate{t_1}}{\mathtt{ap}(\translate{t_2};\alpha)} & \fresh{\alpha}
    \end{array}
  \]
\end{frameddefinition}
\medskip
\begin{example}
  Consider the term $(\lambda x. x*x)\ \natlit{2}$ in \surfacelang.
  We can translate this term and evaluate it in \targetlang\ as follows:
  \begin{align*}
    \cut{\SCcocase{\mathtt{ap}(x,\beta)\Rightarrow \cut{\mu\gamma.*(x,x;\gamma)}{\beta}}}{\mathtt{ap}(\natlit{2};\vartop)}
    \reducesto \cut{\mu\gamma.*(\natlit{2},\natlit{2};\gamma)}{\vartop}
    \reducesto^{\ast} \cut{\natlit{4}}{\vartop}
  \end{align*}
\end{example}

\subsection{Control Operators}
\label{subsec:syntax-exception-handling}

Finally, we add the feature that we used in the motivating example in the introduction: labels and jumps.
We have to extend \surfacelang\ with $\mathbf{label}$ and $\mathbf{goto}$ constructs but since we can translate them locally to $\mu$-bindings we don't have to add anything to \targetlang.
\medskip

\begin{frameddefinition}[Control Operators]
  \[
    \begin{array}{lcl}
      t & \Coloneqq & \ldots \mid \labelterm{\alpha}{t} \mid \jump{t}{\alpha}
    \end{array}
  \]
  \hfill
  \begin{equation*}
    \translate{\labelterm{\alpha}{t}} \coloneq \mu \alpha.\cut{\translate{t}}{\alpha} \qquad
    \translate{\jump{t}{\alpha}} \coloneq \mu \beta.\cut{\translate{t}}{\alpha} \quad \fresh{\beta}
  \end{equation*}
\end{frameddefinition}
\medskip

A term $\labelterm{\alpha}{t}$ binds a covariable $\alpha$ in the term $t$ and thereby provides a location to which a $\mathbf{goto}$ used within $t$ can jump.
Such a $\jump{t}{\alpha}$ takes the location $\alpha$ as an argument, as well as the term $t$ that should be used to continue the computation at the location where $\alpha$ was bound.
It is a bit tricky to write down precisely how the evaluation of $\mathbf{label}$ and $\mathbf{goto}$ works, but the following two rules are a good approximation, where we assume that $\alpha$ does not occur free in $\valueof{t}$:
\begin{equation*}
  \labelterm{\alpha}{\valueof{t}} \reducesto \valueof{t} \qquad\qquad
  \labelterm{\alpha}{\dots \jump{\valueof{t}}{\alpha} \dots} \reducesto \valueof{t}
\end{equation*}
The left rule says that when the labeled term $t$ can be evaluated to a value $\valueof{t}$ without ever using a $\mathbf{goto}$, then we can discard the surrounding $\mathbf{label}$.
The rule on the right says that if we do have a $\mathbf{goto}$ which jumps to the $\mathbf{label}$ $\alpha$ with a value $\valueof{t}$, then we discard everything between the $\mathbf{label}$ and the $\mathbf{goto}$ and continue the computation with this value $\valueof{t}$.
In order to make this second rule precise, we have to make explicit what we only indicate with the ellipses separating the label from the jump; we will do so in \cref{sec:focusing}.

\begin{example}
  \label{ex:control}
  In the introduction, we used the example of a fast multiplication function which multiplies all the elements of a list and short-circuits the computation if it encounters a zero.
  As we have allowed top-level definitions to pass covariables as arguments, we can now write the example of the introduction.
  \begin{align*}
    &\mathbf{def}\ \text{mult}(l) \coloneq \labelterm{\alpha}{\text{mult'}(l;\alpha)} \\
    &\mathbf{def}\ \text{mult'}(l;\alpha) \coloneq \case{l}{ \text{Nil} \Rightarrow 1, \text{Cons}(x,xs) \Rightarrow \ifzero{x}{\jump{0}{\alpha}}{x * \text{mult'}(xs; \alpha)}}
  \end{align*}
  When we translate to \targetlang\ and simplify the resulting term, we get the result:
  \begin{align*}
    &\mathbf{def}\ \text{mult}(l;\alpha) \coloneq \text{mult'}(l;\alpha,\alpha) \\
    &\mathbf{def}\ \text{mult'}(l;\alpha, \beta) \coloneq \\
    & \quad \cut{l}{\SCcase{\text{Nil} \Rightarrow \cut{1}{\beta}, \text{Cons}(x,xs) \Rightarrow \ifzero{x}{\cut{0}{\alpha}}{*(x,\mu \gamma.\text{mult'}(xs;\alpha,\gamma);\beta)}}}
  \end{align*}
  This is almost the result we have seen in the introduction.
  The only difference is that the recursive call to $\text{mult'}$ is nested inside the multiplication.
  This is the same problem we have seen with nested arithmetic operations at the end of \cref{subsec:syntax-arith-expressions} and we will address it in the next section.
\end{example}

The $\mathbf{label}/\mathbf{goto}$ control operator we have introduced in this subsection is of course named after the goto instructions and labels which can be found in many imperative programming languages.
Our adaption to the context of functional programming languages is similar to classical control operators (see \cref{subsec:insights:let-vs-control} for a more precise discussion) such as J \cite{Landin1965} or $\mathtt{let/cc}$ (also known as $\mathbf{escape}$) \cite{Reynolds1972definitional}; the programming language Scala also provides the closely related $\mathtt{boundary}/\mathtt{break}$\footnote{See~\href{https://www.scala-lang.org/api/3.3.0/scala/util/boundary$.html}{scala-lang.org/api/3.3.0/scala/util/boundary\$.html}.} where a boundary marks a block of code to which the programmer can jump with a break instruction.
One central property of this control effect is that it is lexically scoped, since the label names $\alpha$ are passed around lexically and can be shadowed.
This distinguishes them from dynamically scoped control operators like the exception mechanisms found in many programming languages like Java or C++.
(A dynamically scoped variant of our control operator would omit the label names, and the jump in $\mathbf{label}\ \{ \ldots \mathbf{goto}(t) \ldots\}$ would return to the nearest enclosing label at runtime.)
We follow the more recent reappraisal of lexically scoped control effects, for example by \citet{Zhang2016} in the case of exceptions or by \citet{Brachthaeuser2020capabilities} in the case of effect handlers and delimited continuations.

%% file: sections/focusing.tex
At the end of \cref{subsec:syntax-arith-expressions} we ran into the problem that we cannot yet fully evaluate the term $(\natlit{2} * \natlit{4}) + \natlit{5}$ in \surfacelang\ or its translation $\mu \alpha.+(\mu \beta.*(\natlit{2},\natlit{4};\beta),\natlit{5};\alpha)$ in \targetlang\ with the rules that are available to us: we are stuck.
In this section, we finally address this problem.
We are going to show how we can evaluate subexpressions of \surfacelang\ in \cref{subsec:focusing:fun}, but since we are ultimately more interested in compiling programs into \targetlang\ to optimize and reduce those programs, we are spending more time on the problem for \targetlang\ in \cref{subsec:focusing:core}.

\subsection{Evaluation Contexts for Fun}
\label{subsec:focusing:fun}

The problem with evaluating the term $(\natlit{2} * \natlit{4}) + \natlit{5}$ is that the available rules only allow to reduce direct redexes and not redexes that are nested somewhere within a term.
Evaluation contexts solve this problem by specifying the locations within a term which are in evaluation position.
In our example, the term $(\natlit{2} * \natlit{4}) + \natlit{5}$ can be factored into the evaluation context $\square + \natlit{5}$ and the redex $\natlit{2} * \natlit{4}$.
We can then use the old rules to reduce this redex to $\natlit{8}$ and then plug this result back into the evaluation context, which yields the new term $\natlit{8} + \natlit{5}$.
The syntax of evaluation contexts is given in \cref{def:focusing:evaluationcontexts}.
\medskip

\begin{frameddefinition}[Evaluation Contexts]
  \label{def:focusing:evaluationcontexts}
  Evaluation contexts $E$ are defined as:\\
  \[
    \begin{array}{rcl}
      E & \Coloneqq & \square \mid E\odot t \mid \valueof{t} \odot E \mid \ifzero{E}{t}{t} \mid \letin{x}{E}{t} \mid f(\overline{\valueof{t}},E,\overline{t})\mid K(\overline{\valueof{t}},E,\overline{t})\\
        & \mid & \case{E}{\overline{K(\overline{x}) \Rightarrow t}}\mid E\ t \mid \valueof{t}\ E \mid E.D(\overline{t}) \mid \valueof{t}.D(\overline{\valueof{t}},E,\overline{t}) \mid \labelterm{\alpha}{E} \mid \jump{E}{\alpha}
    \end{array}
  \]
\end{frameddefinition}
\medskip

These evaluation contexts also allow us to specify formally the second approximate evaluation rule of the $\mathbf{label}$ and $\mathbf{goto}$ constructs from \cref{subsec:syntax-exception-handling}:
\begin{equation*}
        E[\labelterm{\alpha}{E^{\prime}[\jump{\valueof{t}}{\alpha}]}] \reducesto E[\valueof{t}]
\end{equation*}
Here we again assume that $\alpha$ does not occur free in $\valueof{t}$ and moreover that the inner evaluation context $E^{\prime}$ does not contain another $\mathbf{label}$ construct.
For the full operational semantics of $\mathbf{label}/\mathbf{goto}$ we also need to handle the cases where $\alpha$ can occur free in $\valueof{t}$ and where $E^{\prime}$ can contain other $\mathbf{label}$s.
Otherwise, we could get stuck during evaluation even for closed and well-typed terms, i.e., the progress theorem (see \cref{teo:progress-surface} in \cref{subsec:typing:theorems}) would not hold.
As the full semantics is in essence that of other classical control operators (i.p., $\mathtt{let/cc}$; also see the discussion in \cref{subsec:insights:let-vs-control}) and requires some more formalism, we do not give it here and instead refer the interested reader to the brief discussion in \cref{sec:appendix:semantics-label}.

With evaluation contexts, we finally have a working and precise operational semantics for \surfacelang\ (apart from the approximate rules for $\mathbf{label}$ and $\mathbf{goto}$) which we can use to reason about programs.
Unfortunately, it is wildly inefficient to implement an evaluator which uses evaluation contexts in the way described above.
The reason for this inefficiency is that we very elegantly specified how a term can be factored into an evaluation context and a redex, but the evaluator which implements this behavior has to search for the next redex after every single evaluation step.
We will see in the next section that we have a better solution once our programs are compiled into \targetlang.

\subsection{Focusing on Evaluation in Core}
\label{subsec:focusing:core}

Let us now come back to the problem in \targetlang\ and find a solution for the stuck term $\mu \alpha.+(\mu \beta.*(\natlit{2},\natlit{4};\beta),\natlit{5};\alpha)$.
We know that we have to evaluate $\mu \beta.*(\natlit{2},\natlit{4};\beta)$ next and then somehow plug the intermediate result into the hole $[\cdot]$ in the producer $\mu \alpha.+([\cdot],\natlit{5};\alpha)$.
If we give the intermediate result the name $x$ and play around with cuts, $\mu$-bindings and $\tilde\mu$ bindings, we might discover that we can recombine all these parts in the following way:
\begin{equation*}
  \mu \alpha.\cut{\mu \beta.*(\natlit{2},\natlit{4};\beta)}{\tilde\mu x.+(x,\natlit{5};\alpha)}
\end{equation*}
This term looks a bit mysterious, but the transformation corresponds roughly to what happens when we translate the term $\letin{x}{2*4}{x + 5}$ instead of $(2 * 4) + 5$ into \targetlang.
That is, we have lifted a subcomputation to the outside of the term we are evaluating.
This kind of transformation is called \emph{focusing} \cite{Andreoli1992logicprogramming, Curien2010} and we use it to solve the problem with stuck terms in \targetlang.
We can see that it worked in our example because the term now fully evaluates to its normal form.

\begin{example}
  The producer $\mu \alpha.\cut{\mu \beta.*(\natlit{2},\natlit{4};\beta)}{\tilde\mu x.+(x,\natlit{5};\alpha)}$ reduces as follows:
  \begin{align*}
    \cut{\mu \beta.*(\natlit{2},\natlit{4};\beta)}{\tilde\mu x.+(x,\natlit{5};\vartop)}
    &\reducesto *(\natlit{2},\natlit{4};\tilde\mu x.+(x,\natlit{5};\vartop)) \\
    &\reducesto \cut{\natlit{8}}{\tilde\mu x.+(x,\natlit{5};\vartop)} \\
    &\reducesto +(\natlit{8},\natlit{5};\vartop) \reducesto \cut{\natlit{13}}{\vartop}
  \end{align*}
\end{example}

Once we have settled on focusing, we have another choice to make: Do we want to use this trick during the evaluation of a statement or as a preprocessing step before we start with the evaluation?
These two alternatives are called dynamic and static focusing.

\begin{description}[leftmargin=0.6cm]
    \item[Dynamic Focusing] With dynamic focusing \cite{Wadler2003call} we add additional evaluation rules, usually called $\varsigma$-rules, to lift sub-computations to the outside of the statement we are evaluating.
    \item[Static Focusing] For static focusing \cite{Curien2000duality} we perform a transformation on the code before we start evaluating it.
        This results in a focused normal form which is a subset of the syntax of \targetlang\ that we have described so far.
\end{description}

Dynamic focusing is great for reasoning about the meaning of programs, but static focusing is more efficient if we are interested in compiling and running programs.
For this reason, we only consider static focusing in what follows.

\medskip
\begin{frameddefinition}[Static Focusing]
        \label{def:static-focusing}
        Static focusing is done using the following rules:
        \[
        \begin{array}{r c l}
                \multicolumn{3}{c}{\emph{Producers}}\\[0.1cm]
                \focus{\natlit{n}} & \coloneq & \natlit{n}\\
                \focus{x} & \coloneq & x\\
                \focus{\mu\alpha. s} & \coloneq & \mu\alpha.\focus{s}\\
                \focus{K(\overline{\valueof{p}},p,\overline{p};\overline{c})} & \coloneq &
                        \mu\alpha.\cut{\focus{p}}{\tilde\mu x.\cut{\focus{K(\overline{\valueof{p}},x,\overline{p},\overline{c})}}{\alpha}} \quad \novalue{p}\\
                \mathcal{F}(K(\overline{\valueof{p}};\overline{c})) & \coloneq & K(\overline{\focus{\valueof{p}}};\overline{\focus{c}})\\
                \mathcal{F}(\SCcocase{\overline{D(\overline{x};\overline{\alpha})\Rightarrow s}}) & \coloneq &
                        \SCcocase{\overline{D(\overline{x};\overline{\alpha}) \Rightarrow \focus{s}}}\\[0.1cm]
                \multicolumn{3}{c}{\emph{Consumers}}\\[0.1cm]
                \focus{\alpha} & \coloneq & \alpha\\
                \focus{\tilde\mu x.s} & \coloneq & \tilde\mu x. \focus{s}\\
                \focus{\SCcase{\overline{K(\overline{x};\overline{\alpha})\Rightarrow s}}} & \coloneq &
                        \SCcase{\overline{K(\overline{x};\overline{\alpha})\Rightarrow \focus{s}}}\\
                \focus{D(\overline{\valueof{p}},p,\overline{p},\overline{c})} & \coloneq &
                        \tilde\mu y. \cut{\focus{p}}{\tilde\mu x. \cut{y}{\focus{D(\overline{\valueof{p}},x,\overline{p};\overline{c})}}} \quad \novalue{p}\\
                \focus{D(\overline{\valueof{p}};\overline{c})} & \coloneq & D(\overline{\focus{\valueof{p}}};\overline{\focus{c}})\\[0.1cm]
                \multicolumn{3}{c}{\emph{Statements}}\\[0.1cm]
                \focus{\cut{p}{c}} & \coloneq & \cut{\focus{p}}{\focus{c}}\\
                \focus{\odot(p_1,p_2,c)} & \coloneq & \cut{\focus{p_1}}{\tilde\mu x.\focus{\odot(x,p_2,c)}} \quad \novalue{p_1}\\
                \focus{\odot(\valueof{p},p,c)} & \coloneq & \cut{\focus{p}}{\tilde\mu x.\focus{\odot(\valueof{p},x,c)}} \quad \novalue{p}\\
                \focus{\odot(\valueof{p}_1,\valueof{p}_2,c)} & \coloneq & \odot(\mathcal{F}(\valueof{p}_1),\mathcal{F}(\valueof{p}_2),\focus{c})\\
                \focus{\ifzero{p}{s_1}{s_2}} & \coloneq & \cut{\focus{p}}{\tilde\mu x.\ifzero{x}{s_1}{s_2}} \quad \novalue{p}\\
                \focus{\ifzero{\valueof{p}}{s_1}{s_2}} & \coloneq & \ifzero{\focus{\valueof{p}}}{\focus{s_1}}{\focus{s_2}} \\
                \focus{\text{f}(\overline{\valueof{p}},p,\overline{p};\overline{c})} & \coloneq &
                        \cut{\focus{p}}{\tilde\mu x.\focus{f(\overline{\valueof{p}},x,\overline{p};\overline{c})}} \quad \novalue{p}\\
                \focus{\text{f}(\overline{\valueof{p}},\overline{c})} & \coloneq & \text{f}(\overline{\focus{\valueof{p}}},\overline{\mathcal{F}(c)})
        \end{array}
        \]
\end{frameddefinition}
\medskip

The complete rules for static focusing are presented in \cref{def:static-focusing}.
Most of these rules are only concerned with performing the focusing transformation on all subexpressions, but some of the clauses where something interesting happens are the clauses for binary operators:
\begin{align*}
  \focus{\odot(p_1,p_2,c)} & \coloneq \cut{\focus{p_1}}{\tilde\mu x.\focus{\odot(x,p_2,c)}} \quad \novalue{p_1}\\
  \focus{\odot(\valueof{p},p,c)} & \coloneq \cut{\focus{p}}{\tilde\mu x.\focus{\odot(\valueof{p},x,c)}} \quad \novalue{p}\\
  \focus{\odot(\valueof{p}_1,\valueof{p}_2,c)} & \coloneq \odot(\mathcal{F}(\valueof{p}_1),\mathcal{F}(\valueof{p}_2),\focus{c})
\end{align*}
The first two clauses look for the arguments of the binary operator $\odot$ which are not values and use the trick described above to lift them to the outside.
Focusing is invoked recursively until the binary operator is only applied to values and the third clause comes into play.
This third clause then applies the focusing transformation to all arguments of the binary operator.
The clauses for constructors, destructors, $\mathbf{ifz}$ and calls to top-level definitions work in precisely the same way as those for binary operators.
It is noteworthy that by focusing the producer arguments of destructors we guarantee that the evaluation rule for codata types can fire.
If we had not required the producer arguments to be values in that rule (but only that the destructor is a covalue), we could easily introduce an unfocused term again by substituting a non-value for a variable.

The focusing transformation described in \cref{def:static-focusing} is not ideal since it creates a lot of administrative redexes.
As an example, consider how the statement defining $\text{mult'}$ from \cref{ex:control} is focused:
\begin{align*}
  & \focus{\cut{l}{\SCcase{\text{Nil} \Rightarrow \cut{1}{\beta}, \text{Cons}(x,xs) \Rightarrow \ifzero{x}{\cut{0}{\alpha}}{*(x,\mu \gamma.\text{mult'}(xs;\alpha,\gamma);\beta)}}}} \\
  & = \cut{l}{\SCcase{\text{Nil} \Rightarrow \cut{1}{\beta}, \text{Cons}(x,xs) \Rightarrow \ifzero{x}{\cut{0}{\alpha}}{\cut{\mu \gamma.\text{mult'}(xs;\alpha,\gamma)}{\tilde\mu z.*(x,z;\beta)}}}}
\end{align*}
Focusing has introduced the administrative redex $\cut{\mu \gamma.\text{mult'}(xs;\alpha,\gamma)}{\tilde\mu z.*(x,z;\beta)}$ in the second statement of the $\mathbf{ifz}$.
After reducing this redex to $\text{mult'}(xs;\alpha,\tilde\mu z.*(x,z;\beta))$, we finally arrive at the result from the introduction.
In the implementation, we solve this problem by statically reducing administrative redexes in a simplification step, but it is also possible to come up with a more elaborate definition of focusing which does not create them in the first place.
Such an optimized focusing transformation is, however, much less transparent than the one we have described.

%% file: sections/typing.tex
In this section, we introduce the typing rules for \surfacelang\ in \cref{subsec:typing:fun} and for \targetlang\ in \cref{subsec:typing:core}.
In \cref{subsec:typing:theorems} we state type soundness for both languages and prove that the translation from \surfacelang\ to \targetlang\ preserves the typeability of programs.
We use the same constructors, destructors, types and typing contexts for both \surfacelang\ and \targetlang, which are summarized in \cref{def:typing:types-contexts}.
Note that we distinguish between producer and consumer variables in the typing contexts, which we indicate with the prd and cns annotations.

\medskip
\begin{frameddefinition}[Types and Typing Contexts]
  \[
    \begin{array}{l c l r}
      K      & \Coloneqq & \mathtt{Nil} \mid \mathtt{Cons} \mid \mathtt{Tup} & \emph{Constructors}\\
      D      & \Coloneqq & \mathtt{hd} \mid \mathtt{tl} \mid \mathtt{fst} \mid \mathtt{snd} \mid \mathtt{ap} & \emph{Destructors}\\
      \tau   & \Coloneqq & \tyint \mid \mathtt{List}(\tau) \mid \mathtt{Pair}(\tau,\tau) \mid \mathtt{Stream}(\tau)\ \mid\ \mathtt{LPair}(\tau,\tau) \mid \tau \to \tau & \emph{Types}\\
      \Gamma & \Coloneqq & \emptyset \mid \Gamma, x \prd \tau \mid \Gamma, \alpha \cnt \tau & \emph{Typing Contexts}
    \end{array}
  \]
  \label{def:typing:types-contexts}
\end{frameddefinition}
\medskip

We specialize the rules for data types to the concrete types $\mathtt{Pair}$ and $\mathtt{List}$, and the rules for codata types to $\mathtt{LPair}$, $\mathtt{Stream}$ and functions $\sigma \to \tau$.
A realistic programming language would use type declarations introduced by the programmer to typecheck data and codata types instead of using these special cases.
But the formalization of such a general mechanism for specifying data and codata types makes the typing rules less readable.
This kind of mechanism for specifying algebraic data and codata types in sequent-calculus-based languages can be found in \cite{Downen2015structural} or \cite[section~8]{Downen2020compiling}.
In all of the typing rules below we assume that we have a program environment which contains type declarations for all the definitions contained in the program, but don't explicitly thread this program environment through each of the typing rules.

\subsection{Typing Rules for Fun}
\label{subsec:typing:fun}

We don't discuss the typing rules for \surfacelang\ in detail since they are mostly standard.
Instead, we provide the full rules in \cref{sec:appendix:typing-surfacelang}.
The language \surfacelang\ only has one syntactic category, terms, so we only need one typing judgment $\Gamma \vdash t : \tau$.
This typing judgment says that in the context $\Gamma$ (which contains type assignments for both variables and covariables) the term $t$ has type $\tau$.
The only two interesting rules concern the control operators $\mathbf{label}$ and $\mathbf{goto}$:

\begin{minipage}{0.45\textwidth}
  \begin{prooftree}
    \AxiomC{$\Gamma,\alpha\cnt \tau \vdash t : \tau$}
    \RightLabel{\textsc{Label}}
    \UnaryInfC{$\Gamma \vdash \labelterm{\alpha}{t} : \tau$}
  \end{prooftree}
\end{minipage}
\hfill
\begin{minipage}{0.45\textwidth}
  \begin{prooftree}
    \AxiomC{$\Gamma \vdash t :\tau$}
    \AxiomC{$\alpha \cnt \tau \in \Gamma$}
    \RightLabel{\textsc{Goto}}
    \BinaryInfC{$\Gamma \vdash \jump{t}{\alpha} : \tau'$}
  \end{prooftree}
\end{minipage}
\smallskip

In the rule \textsc{Label} we add the covariable $\alpha\cnt\tau$ to the typing context which is used to typecheck the term $t$.
The labeled expression $\labelterm{\alpha}{t}$ can return in only one of two ways: either the term $t$ is evaluated to a value and returned, or a jump instruction is used to jump to the label $\alpha$.
For this reason, the term $t$ and the label $\alpha$ must have the same type $\tau$, which is also the type for the labeled expression itself.

In the rule \textsc{Goto} we require that the covariable $\alpha$ is in the context with type $\tau$, and that the term $t$ can be typechecked with the same type.
The term $\jump{t}{\alpha}$ itself can be used at any type $\tau'$ because it does not return to its immediately surrounding context.

\subsection{Typing Rules for Core}
\label{subsec:typing:core}

The complete typing rules for \targetlang\ are given in \cref{fig:typing:targetlang}, but we will present them step by step.
We now have producers, consumers and statements as different syntactic categories.
For each of these categories, we use a separate judgment form:

\begin{description}
  \item[Producers] The judgment $\Gamma \vdash p \prd \tau$ says that the producer $p$ has type $\tau$ in context $\Gamma$.
  \item[Consumers] The judgment $\Gamma \vdash c \cnt \tau$ says that the consumer $c$ has type $\tau$ in context $\Gamma$.
  \item[Statements] The judgment $\Gamma \vdash s$ says that the statement $s$ is well-typed in context $\Gamma$. In contrast to producers and consumers, statements do not have a type.
\end{description}
All typing judgments are also implicitly indexed by the program $P$ containing the top-level definitions.
However, as these definitions are only needed when typechecking their calls (rule \textsc{Call}), we usually omit the index from the presentation.

\input{sections/target-typing.tex}

The three different judgments can be illustrated by the rules for variables, covariables and cuts.
In the rules \textsc{Var}$_1$ and \textsc{Var}$_2$ we check that a variable or covariable is contained in the typing context $\Gamma$ and then type the variable as a producer or the covariable as a consumer.
The rule \textsc{Cut} combines a producer $p$ and consumer $c$ of the same type $\tau$ into the statement $\cut{p}{c}$ which does not have a type.

\begin{minipage}{0.25\textwidth}
  \begin{prooftree}
    \AxiomC{$x\prd\tau\in\Gamma$}
    \RightLabel{\textsc{Var}$_1$}
    \UnaryInfC{$\Gamma \vdash x \prd \tau$}
  \end{prooftree}
\end{minipage}
\hfill
\begin{minipage}{0.25\textwidth}
  \begin{prooftree}
    \AxiomC{$\alpha\cnt\tau\in\Gamma$}
    \RightLabel{\textsc{Var}$_2$}
    \UnaryInfC{$\Gamma \vdash \alpha \cnt \tau$}
  \end{prooftree}
\end{minipage}
\hfill
\begin{minipage}{0.4\textwidth}
  \begin{prooftree}
    \AxiomC{$\Gamma \vdash p \prd \tau$}
    \AxiomC{$\Gamma \vdash c \cnt \tau$}
    \RightLabel{\textsc{Cut}}
    \BinaryInfC{$\Gamma \vdash \cut{p}{c}$}
  \end{prooftree}
\end{minipage}
\smallskip

The two unusual constructs which are central to \targetlang\ and give the \lambdamumu\ its name are the $\mu$- and $\tilde\mu$-abstractions.
A $\mu$-abstraction $\mu \alpha.s$ abstracts over a consumer $\alpha$ of type $\tau$ in the statement $s$ and is typed as a producer of type $\tau$.
A $\tilde\mu$-abstraction $\tilde\mu x.s$ abstracts over a producer $x$ of type $\tau$ and is typed as a consumer of type $\tau$, which can be seen in the following two rules.

\begin{minipage}{0.45\textwidth}
  \begin{prooftree}
    \AxiomC{$\Gamma,\alpha \cnt \tau \vdash s$}
    \RightLabel{\textsc{$\mu$}}
    \UnaryInfC{$\Gamma \vdash \mu\alpha.s \prd \tau$}
  \end{prooftree}
\end{minipage}
\hfill
\begin{minipage}{0.45\textwidth}
  \begin{prooftree}
    \AxiomC{$\Gamma,x \prd \tau \vdash s$}
    \RightLabel{\textsc{$\tilde\mu$}}
    \UnaryInfC{$\Gamma\vdash \tilde\mu x.s \cnt \tau$}
  \end{prooftree}
\end{minipage}
\smallskip

\subsubsection{Data and Codata Types}
\cref{fig:typing:targetlang} contains the typing rules for both $\mathtt{Pair}$ and $\mathtt{List}$; since their rules are so similar we only discuss those of $\mathtt{Pair}$ explicitly:

\begin{minipage}{0.35\textwidth}
  \begin{prooftree}
    \AxiomC{$\Gamma \vdash t_1 \prd \tau_1$}
    \AxiomC{$\Gamma \vdash t_2 \prd \tau_2$}
    \RightLabel{\textsc{Tup}}
    \BinaryInfC{$\Gamma \vdash \mathtt{Tup}(t_1,t_2) \prd \mathtt{Pair}(\tau_1,\tau_2)$}
  \end{prooftree}
\end{minipage}
\hfill
\begin{minipage}{0.575\textwidth}
  \begin{prooftree}
    \AxiomC{$\Gamma, x \prd \tau_1, y \prd \tau_2 \vdash s$}
    \RightLabel{\textsc{Case-Pair}}
    \UnaryInfC{$\Gamma \vdash \SCcase{\mathtt{Tup}(x,y)\Rightarrow s} \cnt \mathtt{Pair}(\tau_1,\tau_2)$}
  \end{prooftree}
\end{minipage}
\smallskip

\noindent
In the rule \textsc{Tup} we type a pair constructor $\mathtt{Tup}$ applied to two arguments as a producer, and in the rule \textsc{Case-Pair} we type the case, which pattern-matches on this constructor and brings two variables into scope, as a consumer.

The typing rules for codata types look exactly the same, only the roles of producers and consumers are swapped.

\begin{minipage}{0.3\textwidth}
\begin{prooftree}
  \AxiomC{$\Gamma \vdash k \cnt \tau$}
  \RightLabel{\textsc{Hd}}
  \UnaryInfC{$\Gamma \vdash \mathtt{hd}(k) \cnt \mathtt{Stream}(\tau)$}
\end{prooftree}
\end{minipage}
\hfill
\begin{minipage}{0.65\textwidth}
  \begin{prooftree}
    \AxiomC{$\Gamma, \alpha \cnt \tau \vdash s_1$}
    \AxiomC{$\Gamma, \beta\cnt\mathtt{Stream}(\tau) \vdash s_2$}
    \RightLabel{\textsc{Cc-Str}}
    \BinaryInfC{$\Gamma \vdash \SCcocase{\mathtt{hd}(\alpha)\Rightarrow s_1, \mathtt{tl}(\beta)\Rightarrow s_2} \prd \mathtt{Stream}(\tau)$}
  \end{prooftree}
\end{minipage}
\smallskip

Most of the other rules directly correspond to a similar rule for \surfacelang.
When typing arithmetic expressions, for example, we only have to make sure all subterms have type $\tyint$.

We typecheck programs using the two rules \textsc{Wf-Empty} and \textsc{Wf-Cons}.
The former is used to typecheck an empty program, and the rule \textsc{Wf-Cons} extends a typechecked program with a new top-level definition.
When we typecheck the body of this top-level definition that we are about to add, we extend the program with this definition so that it can refer to itself recursively.

\subsection{Type Soundness}
\label{subsec:typing:theorems}

In this section, we discuss the soundness of the type systems for both \surfacelang\ and \targetlang\ and show that the translation $\translate{-}$ preserves the typeability of terms.
We follow \citet{Wright1994Soundness} in presenting type soundness as the combination of a progress and a preservation theorem.
\begin{theorem}[Progress, \surfacelang]
  \label{teo:progress-surface}
  Let $t$ be a closed term in \surfacelang, such that $\vdash t:\tau$ for some type $\tau$.
  Then either $t$ is a value or there is some term $t^{\prime}$ such that $t\reducesto t^{\prime}$.
\end{theorem}

This can easily be proved with an induction on typing derivations.
Due to the presence of the $\mathbf{label}/\mathbf{goto}$ construct, the standard formulation of the (strong) preservation theorem does not immediately hold for \surfacelang\ (also see the discussion in \cref{sec:appendix:semantics-label}).
The following weak form of preservation can again be easily proved by induction.

\begin{theorem}[(Weak) Preservation, \surfacelang]
  \label{teo:preservation-surface}
  Let $t,t^{\prime}$ be terms in \surfacelang\ such that $t\reducesto t^{\prime}$, $\Gamma$ an environment and $P$ a program such that $\Gamma \xvdash{P} t:\tau$.
  Then there is a type $\tau^{\prime}$ such that $\Gamma \xvdash{P} t^{\prime}:\tau^{\prime}$.
\end{theorem}
The usual strong preservation theorem requires $\tau^{\prime} = \tau$.
But in fact, a slight variation of this strong form can be proved for \surfacelang\ by adapting the technique found in Section 6 in \cite{Wright1994Soundness}.
Thus, strong type soundness still does hold.

Before we can state the analogous theorems for \targetlang, we will need an additional definition as a termination condition for evaluation.
\begin{definition}[Terminal statement]
  If $\valueof{p}$ is a producer value in \targetlang\ and $\vartop$ a covariable which does not appear free in $\valueof{p}$, then $\cut{\valueof{p}}{\vartop}$ is called a \emph{terminal statement}.
\end{definition}
Terminal statements in \targetlang\ have the same role as values in \surfacelang.
Some sequent-calculus-based languages use a special statement $\mathbf{Done}$ instead of terminal statements for this purpose.

\begin{theorem}[Progress, \targetlang]
  \label{teo:progress-target}
  Let $s$ be a focused statement in \targetlang\ such that $\vdash s$.
  Then either $s$ is a terminal statement, or there is some $s^{\prime}$ such that $s\reducesto s^{\prime}$.
\end{theorem}
For this theorem, we require $s$ to be focused, in contrast to \surfacelang, where progress holds for any (well-typed) term.
This is because we have used static focusing for full evaluation in \targetlang.
If we used dynamic focusing instead, this requirement could be dropped, corresponding to using evaluation contexts (dynamic) in \surfacelang, instead of a translation to normal form (static).

The Preservation theorem for \targetlang\ is analogous to \surfacelang.
\begin{theorem}[Preservation, \targetlang]
  Let $s,s^{\prime}$ be statements in \targetlang\ with $s\reducesto s^{\prime}$, $\Gamma$ an environment and $P$ a program.
  If $\Gamma \xvdash{P} s$, then $\Gamma \xvdash{P} s^{\prime}$.
\end{theorem}
This theorem can also be proven with a straightforward induction on typing derivations.
Of course, this preservation theorem does not make any assertion about result types, as statements do not return anything that could be typed.
However, if evaluation starts with a statement $\cut{p}{\vartop}$ where $\vartop$ does not occur free in $p$ and ends in a terminal statement $s$, then $s = \cut{\valueof{p}}{\vartop}$ for some producer value $\valueof{p}$.
This is because no reduction step can introduce a free variable, so the final one must be the same as the initial one.
Hence, by well-typedness, if $p \prd \tau$, then also $\valueof{p} \prd \tau$, because $\vartop \cnt \tau$.

Lastly, we come to an important property of the translation between these languages:
\begin{theorem}[Type Preservation of Translation]
  \label{teo:preservation-translation}
  Let $t$ be a term in \surfacelang, $\Gamma$ an environment and $P$ a program.
  If $\Gamma \xvdash{P} t:\tau$ for some type $\tau$, then $\Gamma \xvdash{\translate{P}} \translate{t} \prd \tau$ where $\translate{P}$ denotes the translation of all definitions in $P$.
\end{theorem}
\begin{proof}
  Most cases are straightforward in the proof which proceeds by a structural induction on the typing derivation.
  The interesting cases are when the typing derivation types a control operator.
  The only rule in \surfacelang\ with $\labelterm{\alpha}{t_1}$ in the conclusion is \textsc{Label}.
  This rule has the premise $\Gamma,\alpha\cnt \tau \vdash t_1 :\tau$, and applying the induction hypothesis gives $\Gamma,\alpha\cnt \tau \vdash \translate{t_1} \prd \tau$.
  Then we can derive $\translate{t} = \mu\alpha.\cut{\translate{t_1}}{\alpha} \prd \tau$:
  \begin{prooftree}
    \AxiomC{(Induction Hypothesis)}
    \noLine
    \UnaryInfC{$\Gamma,\alpha\cnt\tau \vdash \translate{t_1} \prd \tau$}
    \AxiomC{}
    \RightLabel{\textsc{Var}$_2$}
    \UnaryInfC{$\Gamma,\alpha\cnt\tau \vdash \alpha \cnt \tau$}
    \RightLabel{\textsc{Cut}}
    \BinaryInfC{$\Gamma,\alpha\cnt\tau \vdash \cut{\translate{t_1}}{\alpha}$}
    \RightLabel{\textsc{$\mu$}}
    \UnaryInfC{$\Gamma \vdash \mu\alpha.\cut{\translate{t_1}}{\alpha} \prd \tau$}
  \end{prooftree}
  The only rule with $\jump{t_1}{\alpha}$ in the conclusion is \textsc{Goto}, which has premises $\Gamma \vdash t_1 : \tau$ and $\alpha \cnt \tau \in \Gamma$.
  Applying the induction hypothesis gives $\Gamma \vdash \translate{t_1} \prd \tau$, and we can therefore type the translation of the translation as follows (where we implicitly use weakening, which is allowed since $\beta$ is fresh)
  \begin{prooftree}
    \AxiomC{(Induction Hypothesis)}
    \noLine
    \UnaryInfC{$\Gamma,\beta\cnt\sigma\vdash\translate{t_1} \prd \tau$}
    \AxiomC{$\alpha\cnt\tau\in\Gamma,\beta\cnt\sigma$}
    \RightLabel{\textsc{Var}$_2$}
    \UnaryInfC{$\Gamma,\beta\cnt\sigma\vdash\alpha\cnt\tau$}
    \RightLabel{\textsc{Cut}}
    \BinaryInfC{$\Gamma,\beta\cnt\sigma \vdash \cut{\translate{t_1}}{\alpha}$}
    \RightLabel{\textsc{$\mu$}}
    \UnaryInfC{$\Gamma\vdash\mu\beta.\cut{\translate{t_1}}{\alpha} \prd \tau$}
  \end{prooftree}
\end{proof}



%% file: sections/target-typing.tex
\begin{figure}[p]
  \begin{minipage}{0.21\textwidth}
    \begin{prooftree}
      \AxiomC{$\Gamma,\alpha\cnt \tau \vdash s$}
      \RightLabel{\textsc{$\mu$}}
      \UnaryInfC{$\Gamma \vdash \mu\alpha.s\prd\tau$}
    \end{prooftree}
  \end{minipage}
  \hfill
  \begin{minipage}{0.21\textwidth}
    \begin{prooftree}
      \AxiomC{$\Gamma,x\prd\tau \vdash s$}
      \RightLabel{\textsc{$\tilde{\mu}$}}
      \UnaryInfC{$\Gamma \vdash \tilde{\mu}x.s \cnt \tau$}
    \end{prooftree}
  \end{minipage}
  \hfill
   \begin{minipage}{0.21\textwidth}
    \begin{prooftree}
      \AxiomC{$x\prd\tau\in\Gamma$}
      \RightLabel{\textsc{Var$_1$}}
      \UnaryInfC{$\Gamma \vdash x\prd\tau$}
    \end{prooftree}
  \end{minipage}
  \hfill
  \begin{minipage}{0.21\textwidth}
    \begin{prooftree}
      \AxiomC{$\alpha\cnt\tau\in\Gamma$}
      \RightLabel{\textsc{Var$_2$}}
      \UnaryInfC{$\Gamma \vdash \alpha \cnt \tau$}
    \end{prooftree}
  \end{minipage}
  \hfill
  \vspace{1em}

  \begin{minipage}{0.45\textwidth}
    \begin{prooftree}
      \AxiomC{$\Gamma \vdash p \prd\tau$}
      \AxiomC{$\Gamma \vdash c \cnt \tau$}
      \RightLabel{\textsc{Cut}}
      \BinaryInfC{$\Gamma \vdash \cut{p}{c}$}
    \end{prooftree}
  \end{minipage}
  \hfill
  \begin{minipage}{0.45\textwidth}
    \begin{prooftree}
      \AxiomC{$\Gamma \vdash p \prd \tyint$}
      \AxiomC{$\Gamma \vdash s_1$}
      \AxiomC{$\Gamma \vdash s_2$}
      \RightLabel{\textsc{IfZ}}
      \TrinaryInfC{$\Gamma \vdash \ifzero{p}{s_1}{s_2}$}
    \end{prooftree}
  \end{minipage}
  \hfill
  \vspace{1em}

  \begin{minipage}{0.35\textwidth}
    \begin{prooftree}
      \AxiomC{\quad}
      \RightLabel{\textsc{Lit}}
      \UnaryInfC{$\Gamma \vdash \natlit{n} \prd \tyint$}
    \end{prooftree}
  \end{minipage}
  \hfill
  \begin{minipage}{0.6\textwidth}
    \begin{prooftree}
      \AxiomC{$\Gamma \vdash p_1 \prd \tyint$}
      \AxiomC{$\Gamma \vdash p_2 \prd \tyint$}
      \AxiomC{$\Gamma \vdash c \cnt \tyint$}
      \RightLabel{\textsc{binop}}
      \TrinaryInfC{$\Gamma \vdash \odot(p_1,p_2;c)$}
    \end{prooftree}
  \end{minipage}
  \hfill
  \vspace{1em}

  \begin{minipage}{\textwidth}
    \begin{prooftree}
      \AxiomC{$\mathbf{def}\ f(\overline{x_i \prd \tau_i};\overline{\alpha_j \cnt \tau_j}) \in P$}
      \AxiomC{$\overline{\Gamma \vdash p_i \prd \tau_i}$}
      \AxiomC{$\overline{\Gamma \vdash c_j \cnt \tau_j}$}
      \RightLabel{\textsc{Call}}
      \TrinaryInfC{$\Gamma \xvdash{P} f(\overline{p_i};\overline{c_j})$}
    \end{prooftree}
  \end{minipage}
  \hfill
  \vspace{1em}

  \begin{minipage}{\textwidth}
    \begin{prooftree}
      \AxiomC{$\Gamma \vdash s_1$}
      \AxiomC{$\Gamma, x\prd\tau,xs\prd\mathtt{List}(\tau)\vdash s_2$}
      \RightLabel{\textsc{Case-List}}
      \BinaryInfC{$\Gamma \vdash \SCcase{\mathtt{Nil}\Rightarrow s_1,\mathtt{Cons}(x,xs) \Rightarrow s_2} \cnt \mathtt{List}(\tau)$}
    \end{prooftree}
  \end{minipage}
  \hfill
  \vspace{1em}

  \begin{minipage}{0.45\textwidth}
    \begin{prooftree}
      \AxiomC{\quad}
      \RightLabel{\textsc{Nil}}
      \UnaryInfC{$\Gamma \vdash \mathtt{Nil}\prd\mathtt{List}(\tau)$}
    \end{prooftree}
  \end{minipage}
  \hfill
  \begin{minipage}{0.45\textwidth}
    \begin{prooftree}
      \AxiomC{$\Gamma \vdash t_1\prd\tau$}
      \AxiomC{$\Gamma \vdash t_2\prd\mathtt{List}(\tau)$}
      \RightLabel{\textsc{Cons}}
      \BinaryInfC{$\Gamma \vdash \mathtt{Cons}(t_1,t_2)\prd\mathtt{List}(\tau)$}
    \end{prooftree}
  \end{minipage}
  \hfill
  \vspace{1em}

  \begin{minipage}{0.4\textwidth}
    \begin{prooftree}
      \AxiomC{$\Gamma \vdash t_1 \prd \tau_1$}
      \AxiomC{$\Gamma \vdash t_2 \prd \tau_2$}
      \RightLabel{\textsc{Tup}}
      \BinaryInfC{$\Gamma \vdash \mathtt{Tup}(t_1,t_2)\prd\mathtt{Pair}(\tau_1,\tau_2)$}
    \end{prooftree}
  \end{minipage}
  \hfill
  \begin{minipage}{0.55\textwidth}
    \begin{prooftree}
      \AxiomC{$\Gamma, x\prd\tau_1,y\prd\tau_2 \vdash s$}
      \RightLabel{\textsc{Case-Pair}}
      \UnaryInfC{$\Gamma \vdash \SCcase{\mathtt{Tup}(x,y)\Rightarrow s}\cnt\mathtt{Pair}(\tau_1,\tau_2)$}
    \end{prooftree}
  \end{minipage}
  \hfill
  \vspace{1em}

  \begin{minipage}{0.45\textwidth}
    \begin{prooftree}
      \AxiomC{$\Gamma \vdash k\cnt\tau$}
      \RightLabel{\textsc{Hd}}
      \UnaryInfC{$\Gamma \vdash \mathtt{hd}(k)\cnt \mathtt{Stream}(\tau)$}
    \end{prooftree}
  \end{minipage}
  \hfill
  \begin{minipage}{0.45\textwidth}
    \begin{prooftree}
      \AxiomC{$\Gamma \vdash k\cnt\mathtt{Stream}(\tau)$}
      \RightLabel{\textsc{Tl}}
      \UnaryInfC{$\Gamma \vdash \mathtt{tl}(k)\cnt \mathtt{Stream}(\tau)$}
    \end{prooftree}
  \end{minipage}
  \hfill
  \vspace{1em}

  \begin{minipage}{\textwidth}
    \begin{prooftree}
      \AxiomC{$\Gamma,\alpha\cnt \tau \vdash s_1$}
      \AxiomC{$\Gamma,\beta\cnt\mathtt{Stream}(\tau)\vdash s_2$}
      \RightLabel{\textsc{Cocase-Stream}}
      \BinaryInfC{$\Gamma \vdash \SCcocase{\mathtt{hd}(\alpha)\Rightarrow s_1,\mathtt{tl}(\beta)\Rightarrow s_2} \prd \mathtt{Stream}(\tau)$}
    \end{prooftree}
  \end{minipage}
  \hfill
  \vspace{1em}

  \begin{minipage}{0.45\textwidth}
    \begin{prooftree}
      \AxiomC{$\Gamma \vdash k\cnt \tau_1$}
      \RightLabel{\textsc{Fst}}
      \UnaryInfC{$\Gamma \vdash \mathtt{fst}(k)\cnt\mathtt{LPair}(\tau_1,\tau_2)$}
    \end{prooftree}
  \end{minipage}
  \hfill
  \begin{minipage}{0.45\textwidth}
    \begin{prooftree}
      \AxiomC{$\Gamma \vdash k\cnt\tau_2$}
      \RightLabel{\textsc{Snd}}
      \UnaryInfC{$\Gamma \vdash \mathtt{snd}(k)\cnt \mathtt{LPair}(\tau_1,\tau_2)$}
    \end{prooftree}
  \end{minipage}
  \hfill
  \vspace{1em}

  \begin{minipage}{\textwidth}
    \begin{prooftree}
      \AxiomC{$\Gamma, \alpha \cnt\tau_1 \vdash s_1$}
      \AxiomC{$\Gamma, \beta \cnt\tau_2 \vdash s_2$}
      \RightLabel{\textsc{Cocase-LPair}}
      \BinaryInfC{$\Gamma \vdash \SCcocase{\mathtt{fst}(\alpha)\Rightarrow s_1, \mathtt{snd}(\beta)\Rightarrow s_2} \prd \mathtt{LPair}(\tau_1,\tau_2)$}
     \end{prooftree}
  \end{minipage}
  \hfill
  \vspace{1em}

  \begin{minipage}{0.4\textwidth}
    \begin{prooftree}
      \AxiomC{$\Gamma \vdash p\prd\sigma$}
      \AxiomC{$\Gamma \vdash c\cnt\tau$}
      \RightLabel{\textsc{Ap}}
      \BinaryInfC{$\Gamma \vdash \mathtt{ap}(p,c) \cnt \sigma\to\tau$}
    \end{prooftree}
  \end{minipage}
  \hfill
  \begin{minipage}{0.55\textwidth}
    \begin{prooftree}
      \AxiomC{$\Gamma, x\prd\sigma,\alpha\cnt\tau \vdash s$}
      \RightLabel{\textsc{Cocase-Fun}}
      \UnaryInfC{$\Gamma \vdash \SCcocase{\mathtt{ap}(x,\alpha) \Rightarrow s}\prd\sigma\to\tau$}
    \end{prooftree}
  \end{minipage}
  \hfill
  \vspace{1.5em}
  \rule{\textwidth}{0.4pt}
  \par
  \vspace{1em}
  \begin{minipage}{0.25\textwidth}
    \begin{prooftree}
      \AxiomC{}
      \RightLabel{\textsc{Wf-Empty}}
      \UnaryInfC{$\vdash \wellformed{\emptyset}$}
    \end{prooftree}
  \end{minipage}
  \hfill
  \begin{minipage}{0.7\textwidth}
    \begin{prooftree}
      \AxiomC{$\vdash \wellformed{P}$}
      \AxiomC{$\overline{x\prd\tau_i},\overline{\alpha\cnt\tau_j} \xvdash{P,\mathbf{def}\ \text{f}(\overline{x_i\prd\tau_i};\overline{\alpha_j\cnt\tau_j})\coloneq s} s$}
      \RightLabel{\textsc{Wf-Cons}}
      \BinaryInfC{$\vdash \wellformed{P,\mathbf{def}\ \text{f}(\overline{x_i\prd\tau_i},\overline{\alpha_j\cnt\tau_j}) \coloneq s}$}
    \end{prooftree}
  \end{minipage}
  \hfill
  \caption{Typing rules of \targetlang.}
  \label{fig:typing:targetlang}
\end{figure}

%% file: sections/insights.tex
In the previous section, we have explained \emph{what} the \lambdamumu\ is, and \emph{how} it works.
Now that we know the what and how we can explain \emph{why} this calculus is so interesting.
This section is therefore a small collection of independent insights.
To be clear, these insights are obvious to those who are deeply familiar with the \lambdamumu, but
we can still recall how surprising they were for us when we first learned about them.

\subsection{Evaluation Contexts are First Class}
\label{subsec:insigts:evaluation-context-first-class}
A central feature of the \lambdamumu\ is the treatment of evaluation contexts as first-class objects, as we have mentioned before.
For example, consider
the term $(\natlit{2}*\natlit{3})*\natlit{4}$ in \surfacelang.
When we want to evaluate this, we have to use the evaluation context $\square*\natlit{4}$ to evaluate the subterm $(\natlit{2}*\natlit{3})$ and get $\natlit{6}*\natlit{4}$ which we can then evaluate to $\natlit{24}$.
Translating this term into \targetlang\ gives $\mu\alpha.*(\mu\beta.*(\natlit{2},\natlit{3};\beta),\natlit{4};\alpha)$.
To evaluate this term, we first need to focus it giving
\begin{align*}
  \mu\alpha. \cut{\mu\beta.*(\natlit{2},\natlit{3};\beta)}{\tilde\mu x. *(x,\natlit{4};\alpha)}
\end{align*}
When we now start evaluating with $\vartop$, the steps are the same as in \surfacelang.
Using call-by-value, the $\mu$-abstraction is evaluated first, giving $*(\natlit{2},\natlit{3};*(\tilde\mu x. *(x,\natlit{4};\vartop))$.
This now has the form where the product can be evaluated to $\cut{\natlit{6}}{\tilde\mu x. *(x,\natlit{4};\vartop)}$, after which $\natlit{6}$ is substituted for $x$.
The term $*(\natlit{6},\natlit{4};\vartop)$ can then be directly evaluated to $\natlit{24}$.

After focusing, we can see how $\beta$ is a variable that stands for the evaluation context in \surfacelang.
The term $\tilde\mu x. *(x,\natlit{4};\alpha)$ is the first-class representation of the evaluation context
$\square*\natlit{4}$.
We first evaluate the subexpression $*(\natlit{2},\natlit{3};\beta)$ and then insert the result into $*(x,\natlit{4};\vartop)$ to finish the evaluation, as we did in \surfacelang.
In other words, the $\square$ of an evaluation context in \surfacelang, corresponds to a continuation $\beta$ in \targetlang, and similarly determines in which order subexpressions are evaluated.

\subsection{Data is Dual to Codata}
The sequent calculus clarifies the relation between data and codata as being exactly dual to each other.
When looking at the typing rules in \cref{fig:typing:targetlang}, we can see that data and codata
types are completely symmetric. The two are not symmetric in languages based on natural deduction: A pattern
match on data types includes the scrutinee but there is no corresponding object in the construction of
codata. Similarly, invoking a destructor $D$ of a codata type always includes the codata object $x$ to be destructed, e.g., $x.D(\ldots)$, whereas the invocation of the constructor of a data type has no corresponding object.

This asymmetry is fixed in the sequent calculus. Destructors (such as $\mathtt{fst}$) are first-class and don't require
a scrutinee, which repairs the symmetry to constructors. Similarly, pattern matches ($\SCcase{\ldots}$) do not require an object to destruct, which makes them completely symmetrical to copattern matches.
This duality reduces the conceptual complexity and opens the door towards shared design and implementation of features
of data and codata types.

\subsection{Let-Bindings are Dual to Control Operators}
\label{subsec:insights:let-vs-control}

The $\mathbf{label}$ construct in \surfacelang\ is translated to a $\mu$-binding in \targetlang.
Also, when considering the typing rule for $\labelterm{\alpha}{t}$ in \cref{subsec:typing:fun}, we can see that it directly corresponds to typing a $\mu$-binding with the label $\alpha$ being the bound covariable.
Similarly, a $\mathbf{let}$-binding is translated to a $\tilde\mu$-binding and typing a $\mathbf{let}$-binding in \surfacelang\ closely corresponds to typing a $\tilde\mu$-term in \targetlang.
This way, $\mathbf{label}$s and $\mathbf{let}$-bindings are dual to each other, the same way $\mu$ and $\tilde\mu$ are.
The duality can be extended to other control operators such as $\mathtt{call/cc}$.

As it turns out, the $\mathbf{label}$ construct is very closely related to $\mathtt{call/cc}$.
There are in fact only two differences.
First, $\labelterm{\alpha}{t}$ has the binder $\alpha$ for the continuation built into the construct, just as the variation of $\mathtt{call/cc}$ named $\mathtt{let/cc}$ (which \citet{Reynolds1972definitional} called $\mathbf{escape}$).
The second, and more important difference is that the invocation of the continuation captured by $\labelterm{\alpha}{t}$ happens through an explicit language construct $\jump{t}{\alpha}$.
This makes it easy to give a translation to \targetlang\ as we can simply insert another $\mu$-binding to discard the remaining continuation at exactly the place where the captured continuation is invoked.
In contrast, with $\mathtt{call/cc}$ and $\mathtt{let/cc}$ the continuation is applied in the same way as a normal function, making it necessary to redefine the variable the captured continuation is bound to when translating to \targetlang.
This obscures the duality to $\mathbf{let}$-bindings which is so evident for $\mathbf{label}$ and $\mathbf{goto}$.

To see this, here is a translation of $\mathtt{let/cc}~k~t$ to \targetlang
\begin{align*}
  \translate{\mathtt{let/cc}~k~t}
  &\coloneq
  \mu \alpha.\cut{\SCcocase{ \mathtt{ap}(x, \beta) \Rightarrow \cut{x}{\alpha} }}{\tilde\mu k.\cut{\translate{t}}{\alpha}}
\end{align*}
The essence of the translation still is that the current continuation is captured by the outer $\mu$ and bound to $\alpha$.
But now we also have to transform this $\alpha$ into a function (the $\mathbf{cocase}$ here) which discards its context (here bound to $\beta$) and bind this function to $k$, which is done using $\tilde\mu$.
For $\mathtt{call/cc}$, the duality is even more obscured, as there the binder for the continuation is hidden in the function which $\mathtt{call/cc}$ is applied to.
For the translation, this function must then be applied to the above $\mathbf{cocase}$ and the captured continuation $\alpha$, resulting in the following term (cf. also \cite{Miquey2019}).
\begin{align*}
  \translate{\mathtt{call/cc}~f}
  &\coloneq
  \mu \alpha.\cut{\translate{f}}{\mathtt{ap}(\SCcocase{ \mathtt{ap}(x, \beta) \Rightarrow \cut{x}{\alpha} }, \alpha)}
\end{align*}

Other control operators for undelimited continuations can be translated in a similar way.
For example, consider Felleisen's $\mathcal{C}$ \cite{Felleisen1987syntactic}.
The difference to $\mathtt{call/cc}$ is that $\mathcal{C}$ discards the current continuation if it is not invoked somewhere in the term $\mathcal{C}$ is applied to, whereas $\mathtt{call/cc}$ leaves it in place and thus behaves as a no-op if the captured continuation is never invoked.
The only change that needs to be made in the translation to \targetlang\ is that the top-level continuation $\vartop$ has to be used for the outer cut instead of using the captured continuation.
This is most easily seen for a variation of $\mathcal{C}$ which has the binder for the continuation built into the operator and where the invocation of the continuation is explicit, similar to $\mathbf{label}/\mathbf{goto}$.
Calling this variation $\mathbf{label}_{\mathcal{C}}$, we obtain the following translation
\begin{align*}
  \translate{\mathbf{label}_{\mathcal{C}}\ \alpha\ \{ t \}}
  &\coloneq
  \mu \alpha.\cut{\translate{t}}{\vartop}
\end{align*}
Here the duality to $\mathbf{let}$-bindings is evident again.
The translation for $\mathcal{C}$ itself is then obtained in the same way as for $\mathtt{call/cc}$
\begin{align*}
  \translate{\mathcal{C}~f}
  &\coloneq
  \mu \alpha.\cut{\translate{f}}{\mathtt{ap}(\SCcocase{ \mathtt{ap}(x, \beta) \Rightarrow \cut{x}{\alpha} }, \vartop)}
\end{align*}

\subsection{The Case-of-Case Transformation}
\label{subsec:insights:case-of-case}

One important transformation in functional compilers is the case-of-case transformation.
\citet{Maurer2017} give the following example of this transformation.
The term
\begin{equation*}
    \ite{(\ite{e_1}{e_2}{e_3})}{e_4}{e_5}
\end{equation*}
can be replaced by the term
\begin{equation*}
    \ite{e_1}{(\ite{e_2}{e_4}{e_5})}{(\ite{e_3}{e_4}{e_5})}.
\end{equation*}
Logicians call these kinds of transformations \emph{commutative conversions}, and they play an important role in the study of the sequent calculus.
But as \citet{Maurer2017} show, they are also important for compiler writers who want to generate efficient code.

In the \lambdamumu, commuting conversions don't have to be implemented as a special compiler pass.
They fall out \emph{for free} as a special instance of $\mu$-reductions!
Let us illustrate this point by translating \citeauthor{Maurer2017}'s example into the \lambdamumu.
First, let us translate the two examples using pattern-matching syntax:
\begin{align*}
    \case{(\case{e_1}{\mathtt{T} \Rightarrow e_2; \mathtt{F} \Rightarrow e_3})}{\mathtt{T} \Rightarrow e_4; \mathtt{F} \Rightarrow e_5} \\
    \case{e_1}{\mathtt{T} \Rightarrow \case{e_2}{\mathtt{T} \Rightarrow e_4; \mathtt{F} \Rightarrow e_5}; \mathtt{F} \Rightarrow \case{e_3}{\mathtt{T} \Rightarrow e_4; \mathtt{F} \Rightarrow e_5}}
\end{align*}
Let us now translate these two terms into the \lambdamumu:
\begin{align*}
  &\mu \alpha.\underline{\cut{\mu \beta.\cut{\translate{e_1}}{\SCcase{\mathtt{T} \Rightarrow \cut{\translate{e_2}}{\beta}; \mathtt{F} \Rightarrow \cut{e_3}{\beta}}}}{\SCcase{\mathtt{T} \Rightarrow \cut{\translate{e_4}}{\alpha}, \mathtt{F} \Rightarrow \cut{\translate{e_5}}{\alpha}}}} \\
  & \mu \alpha. \langle \translate{e_1} | \mathbf{case}\ \{ \\
  & \qquad \mathtt{T} \Rightarrow \underline{\cut{\mu \beta.\cut{\translate{e_2}}{\SCcase{\mathtt{T} \Rightarrow \cut{\translate{e_4}}{\beta},\mathtt{F} \Rightarrow \cut{\translate{e_5}}{\beta}}}}{\alpha}} \\
  & \qquad \mathtt{F} \Rightarrow \underline{\cut{\mu \beta.\cut{\translate{e_3}}{\SCcase{\mathtt{T} \Rightarrow \cut{\translate{e_4}}{\beta},\mathtt{F} \Rightarrow \cut{\translate{e_5}}{\beta}}}}{\alpha}} \} \rangle \\
\end{align*}
We can see that just by reducing all of the underlined redexes we reduce both of these examples to the same term.

\subsection{Direct and Indirect Consumers}
\label{subsec:insights:direct-indirect-consumers}
As mentioned in the introduction, a natural competitor of sequent calculus as an intermediate representation is continuation-passing style (CPS).
In CPS, reified evaluation contexts are represented by functions.
This makes the resulting types of programs in CPS arguably harder to understand.
There is, however, another advantage of sequent calculus over CPS as described by \citet{Downen2016sequent}.
The first-class representation of consumers in sequent calculus allows us to distinguish between two different kinds of consumers: direct consumers, i.e., destructors, and indirect consumers.
In particular, this allows to chain direct consumers in \targetlang\ in a similar way as in \surfacelang.

Suppose we have a codata type with destructors $\mathtt{get}$ and $\mathtt{set}$ for getting and setting the value of a reference.
Now consider the following chain of destructor calls on a reference $r$ in \surfacelang
\begin{align*}
  r.\mathtt{set}(3).\mathtt{set}(4).\mathtt{get}()
\end{align*}
A compiler could use a user-defined custom rewrite rule for rewriting two subsequent calls to $\mathtt{set}$ into only the second call.
In \targetlang\ the above example looks as follows:
\begin{align*}
  \mu \alpha. \cut{r}{\mathtt{set}(3; \mathtt{set}(4; \mathtt{get}(\alpha)}
\end{align*}
We still can immediately see the direct chaining of destructors and thus apply essentially the same rewrite rule.
In CPS, however, the example would rather become
\begin{align*}
  \lambda k.\ r.\mathtt{set}(3; \lambda s.\ s.\mathtt{set}(4; \lambda t.\ t.\mathtt{get}(k)))
\end{align*}
The chaining of the destructors becomes obfuscated by the indirections introduced by representing the continuations for each destructor as a function.
To apply the custom rewrite rule mentioned above, it is necessary to see through the lambdas, i.e. the custom rewrite rule has to be transformed to be applicable.

\subsection{Call-By-Value, Call-By-Name and Eta-Laws}
\label{subsec:insights:strict-vs-lazy}

In \cref{subsec:syntax-let-bindings} we already pointed out the existence of statements $\cut{\mu \alpha.s_1}{\tilde\mu x.s_2}$ which are called \emph{critical pairs} because they can a priori be reduced to either $s_1[\tilde\mu x.s_2/\alpha]$ or $s_2[\mu \alpha.s_1/x]$.
These critical pairs were already discussed by \citet{Curien2000duality} when they introduced the \lambdamumu.
One solution is to pick an evaluation order, either call-by-value (cbv) or call-by-name (cbn), that determines to which of the two statements we should evaluate, and in this paper we chose to always use the call-by-value evaluation order.
The difference between these two choices has also been discussed by \citet{Wadler2003call}.
Note that this freedom for the evaluation strategy is another advantage of sequent calculus over continuation-passing style, as the latter always fixes an evaluation strategy.

Which evaluation order we choose has an important consequence for the optimizations we are allowed to perform in the compiler.
If we choose call-by-value, then we are not allowed to use all $\eta$-equalities for codata types, and if we use call-by-name, then we are not allowed to use all $\eta$-equalities for data types.
Let us illustrate the problem in the case of codata types with the following example:
\[
    \cut{\SCcocase{\mathtt{ap}(x;\alpha)\Rightarrow \cut{\mu\beta.s_1}{\mathtt{ap}(x;\alpha)}}}{\tilde\mu x.s_2}
    \equiv_{\eta} \cut{\mu\beta.s_1}{\tilde\mu x.s_2}
\]
We assume that $x$ and $\alpha$ do not appear free in $s_1$.
The $\eta$-transformation is just the ordinary $\eta$-law for functions but applied to the representation of functions as codata types.
The statement on the left-hand side reduces the $\tilde\mu$ first under both call-by-value and call-by-name evaluation order, i.e.
\[
    \cut{\SCcocase{\mathtt{ap}(x;\alpha)\Rightarrow \cut{\mu\beta.s_1}{\mathtt{ap}(x;\alpha)}}}{\tilde\mu x.s_2}
    \begin{array}{l r}
      \reducesto_{\text{cbv}} & s_2[\SCcocase{\dots}/x] \\
      \reducesto_{\text{cbn}} & s_2[\SCcocase{\dots}/x]
    \end{array}
\]
The right-hand side of the $\eta$-equality, however, reduces the $\mu$ first under call-by-value evaluation order, i.e.
\[
    \cut{\mu\beta.s_1}{\tilde\mu x.s_2}
    \begin{array}{l r}
      \reducesto_{\text{cbv}} & s_1[\tilde\mu x.s_2/\beta] \\
      \reducesto_{\text{cbn}} & s_2[\mu\beta. s_1/x]
    \end{array}
\]
Therefore, the $\eta$-equality is only valid under call-by-name evaluation order.
This example shows that the validity of applying this $\eta$-rule as an optimization depends on whether the language uses call-by-value or call-by-name.
If we instead used a data type such as $\mathtt{Pair}$, a similar $\eta$-reduction would only give the same result as the original statement when using call-by-value.

\subsection{Linear Logic and the Duality of Exceptions}
\label{subsec:related-work:linear-logic}

We have introduced the data type $\mathtt{Pair}(\sigma,\tau)$ and the codata type $\mathtt{LPair}(\sigma,\tau)$ as two different ways to formalize tuples.
The data type $\mathtt{Pair}(\sigma,\tau)$ is defined by the constructor $\mathtt{Tup}$ whose arguments are evaluated eagerly, so this type corresponds to strict tuples in languages like ML or OCaml.
The codata type $\mathtt{LPair}(\sigma,\tau)$ is a lazy pair which is defined by its two projections fst and snd, and only when we invoke the first or second projection do we start to compute its contents.
This is closer to how tuples behave in a lazy language like Haskell.

Linear logic \cite{Girard1987,Wadler1990lineartypes} adds another difference to these types.
In linear logic we consider arguments as resources which cannot be arbitrarily duplicated or discarded; every argument to a function has to be used exactly once.
If we follow this stricter discipline, then we have to distinguish between two different types of pairs: In order to use a pair $\sigma \otimes \tau$ (pronounced \enquote{times} or \enquote{tensor}), we have to use both the $\sigma$ and the $\tau$, but if we want to use a pair $\sigma \with \tau$ (pronounced \enquote{with}), we must choose to either use the $\sigma$ or the $\tau$.
It is now easy to see that the type $\sigma \otimes \tau$ from linear logic corresponds to the data type $\mathtt{Pair}(\sigma,\tau)$, since when we pattern match on this type we get \emph{two} variables in the context, one for $\sigma$ and one for $\tau$.
The type $\sigma \with \tau$ similarly corresponds to the type $\mathtt{LPair}(\sigma,\tau)$ which we use by invoking either the first or the second projection, consuming the whole pair.

In addition to these two different kinds of conjunction, we also have two different kinds of disjunction.
These two disjunctions are written $\sigma \oplus \tau$ (pronounced \enquote{plus}) and $\sigma \parr \tau$ (pronounced \enquote{par}) and correspond to two different ways to handle errors in programming languages.
Their typing rules in \targetlang\ are:

\begin{minipage}{0.45\textwidth}
  \begin{prooftree}
    \AxiomC{$\Gamma \vdash t \prd \sigma$}
    \UnaryInfC{$\Gamma \vdash \text{Inl}(t) \prd \sigma \oplus \tau$}
  \end{prooftree}
\end{minipage}
\begin{minipage}{0.45\textwidth}
  \begin{prooftree}
    \AxiomC{$\Gamma \vdash t \prd \tau$}
    \UnaryInfC{$\Gamma \vdash \text{Inr}(t) \prd \sigma \oplus \tau$}
  \end{prooftree}
\end{minipage}

\begin{prooftree}
  \AxiomC{$\Gamma, x \prd \sigma \vdash s_1$}
  \AxiomC{$\Gamma, y \prd \tau \vdash s_2$}
  \BinaryInfC{$\Gamma \vdash \SCcase{\text{Inl}(x) \Rightarrow s_1, \text{Inr}(y) \Rightarrow s_2} \cnt \sigma \oplus \tau$}
\end{prooftree}

\begin{minipage}{0.45\textwidth}
  \begin{prooftree}
    \AxiomC{$\Gamma \vdash c_1 \cnt \sigma$}
    \AxiomC{$\Gamma \vdash c_2 \cnt \tau$}
    \BinaryInfC{$\Gamma \vdash \text{Par}(c1,c2) \cnt \sigma \parr \tau$}
  \end{prooftree}
\end{minipage}
\begin{minipage}{0.45\textwidth}
  \begin{prooftree}
    \AxiomC{$\Gamma, \alpha \cnt \sigma, \beta \cnt \tau \vdash s$}
    \UnaryInfC{$\Gamma \vdash  \SCcocase{ \text{Par}(\alpha,\beta) \Rightarrow s } \prd \sigma \parr \tau$}
  \end{prooftree}
\end{minipage}
\smallskip

Languages like Rust and Haskell use $\sigma \oplus \tau$ for error handling, which corresponds to the \enquote{Either} and \enquote{Result} types in those languages.
This corresponds to the calling convention that the function returns a tagged result which indicates whether an error has occurred or not, and the caller of the function has to check this tag.
The type $\sigma \parr \tau$ behaves differently: A function which returns a value of type $\sigma \parr \tau$ has to be called with two continuations, one for the possibility that the function returns successfully and one for the possibility that the function throws an error.
And the function itself decides which continuation to call, so there is no overhead for checking the result of a function call.
This is quite similar to how some functions in Javascript are called with an \enquote{onSuccess} continuation and an \enquote{onFailure} continuation and different to the exception model of, e.g., Java, where the exception handler is dynamically scoped instead of lexically passed as an argument.
This duality between the two different ways of handling exceptions can be seen most clearly in the sequent calculus; more details on this duality can be found in section 3.4 of \cite{Spiwack2014} or in section 7.1 of \cite{icfp2022}.

%% file: sections/related-work.tex
The central ideas of the calculi that we have presented in this pearl are not novel: the \lambdamumu\ is by now over 20 years old.
We chose a variant of this calculus that can be used as a starting point to explore all the variants that have been described in the literature.
This related work section is therefore intended to provide suggestions for further reading and the chance to dive deeper into specific topics that we have only touched upon.

\subsection{The Sequent Calculus}
\label{subsec:related-work:sequent-calculus}
The basis of our language \targetlang\ is a term assignment system for the sequent calculus, an alternative logical system to natural deduction.
The sequent calculus was originally introduced by Gentzen in the articles \citet{Gentzen1935a,Gentzen1935b,Gentzen1969}. 
For a more thorough introduction to the sequent calculus as a logical system, we can recommend the books by \citet{Negri2001structural} and \citet{TroelstraSchwichtenberg2000} which introduce the sequent calculus and show how it differs from the natural deduction systems that are more commonly taught.

\subsection{Term Assignment for the Sequent Calculus}
\label{subsec:related-work:term-assignment}
The original article which introduced the \lambdamumu\ as a term assignment system for the sequent calculus was by \citet{Curien2000duality}.
Before we list some of the other articles, we should preface them with the following remark on notation:

\begin{remark}[Alternative Notation]
  Our notation for producers, consumers and statements follows the established conventions in the literature.
  However, we diverge in the way that we write typing judgments from the example of \citet{Curien2000duality} which is followed by most other authors.
  We use one typing context $\Gamma$ which binds both variables $x \prd \tau$ and covariables $\alpha \cnt \tau$, whereas \citet{Curien2000duality} use two contexts; a context $\Gamma$ which contains bindings for all variables and a context $\Delta$ which contains the bindings for all covariables.
  The following table summarizes the difference between their notation and the notation used in our paper.
  \medskip

  \begin{center}
    \begin{tabular}{ccc}
      \toprule
      Judgment Form & Our notation & \citet{Curien2000duality} \\
      \midrule
      Typing Producers & $\Gamma \vdash p \prd \tau$ & $\Gamma \vdash p : \tau \mid \Delta$ \\
      Typing Consumers & $\Gamma \vdash c \cnt \tau$ & $\Gamma \mid c : \tau \vdash \Delta$ \\
      Typing Statements & $\Gamma \vdash s$ & $s : (\Gamma \vdash \Delta)$ \\
      \bottomrule
    \end{tabular}
  \end{center}
  \medskip

  The reasons for this divergence are easily explained.
  The notation of \citet{Curien2000duality} with its two contexts $\Gamma$ and $\Delta$ perfectly illustrates the correspondence to the sequent calculus which operates with sequents $\Gamma \vdash \Delta$ which contain multiple formulas on the left- and right-hand side of the turnstile.
  This close correspondence to the sequent calculus is less important for us.
  We found that splitting the context in this way often makes it more difficult to write down rules in their full generality when we extend the language with other features.
  Features which introduce a dependency of later bindings on earlier bindings within a typing context, for example when we add parametric polymorphism, don't fit easily into the format of \citet{Curien2000duality}.
\end{remark}

With these remarks out of the way, we can recommend the articles by \citet{Zeilberger2008unity}, \citet{Downen2014duality, Downen2018, Downen2020compiling}, \citet{MunchMaccagnoni2009} and \citet{Spiwack2014} which were very helpful to us when we learned about the \lambdamumu.
\subsection{Codata Types}
\label{subsec:related-work:codata-types}

Codata types were originally invented by \citet{Hagino1989codatatypes}.
They had the most success in proof assistants such as Agda where they help circumvent certain technical problems that arise when we try to model coinductive types.
Copattern matching as a way to create producers of codata types was popularized by \citet{Abel2013copatterns}, although the basic idea of the concept had been around before that, see, e.g., \cite{Zeilberger2008unity}.
But probably the best starting point to learn more about codata types is an article written by \citet{Downen2019codata}.

\subsection{Control Operators and Classical Logic}
\label{subsec:related-work:control-operators}

The $\mathbf{label}/\mathbf{goto}$ construct that we are using in \surfacelang\ is an example of a control operator, of which Landin's operator J \cite{Thielecke1998,Felleisen1987,Landin1965} likely is the oldest.
Their translation into \targetlang\ uses $\mu$-abstractions, which are also a form of control operator that was originally introduced by \citet{Parigot1992} before it became a part of the \lambdamumu\ of \citet{Curien2000duality}.
Control operators have an important relationship to classical logic via the Curry-Howard isomorphism.
This relationship was discovered by \citet{Griffin1989formulae}; a more thorough introduction can be found in \citet{SorensenUrzyczyn2006}.

\subsection{Different Evaluation Orders}
\label{subsec:related-work:evaluation-orders}

We have already talked about the evaluation strategies call-by-value and call-by-name, and how their difference can be explained by different choices of how a critical pair should be evaluated.
This duality between call-by-value and call-by-name has already been observed by \citet{Filinski1989} and has been explored in more detail by \citet{Wadler2003call, Wadler2005}.
We have also seen in \cref{subsec:insights:strict-vs-lazy} how $\eta$-reduction only works with data types in call-by-value and with codata types in call-by-name.
A lot of people therefore conclude that the choice of an evaluation order should maybe not be a global decision, but should instead depend on the type.
This approach requires tracking the polarity of types and providing additional shift connectives which help mediate between the different evaluation orders; the article by \citet{Downen2018beyondPolarity} is a good entry point for pursuing these kinds of questions which are discussed in detail in \cite{Zeilberger2009} and \cite{Munch2013phd}.
A well-known example of mixing evaluation orders is the call-by-push-value paradigm \cite{Levy1999} which distinguishes value types and computation types and subsumes both call-by-value and call-by-name.

%% file: sections/conclusion.tex
In this functional pearl, we have presented the \lambdamumu\ in the way we introduce it to our colleagues and students on the whiteboard; by compiling small examples of functional programs.
We think this is a better way to introduce programming-language enthusiasts and compiler writers to the \lambdamumu, since it doesn't require prior knowledge of the sequent calculus.
We have also shown \emph{why} we are excited about this calculus, by giving examples of how it allows us to express aspects like strict vs. lazy evaluation or compiler optimizations like case-of-case in an extremely clear way.
We want to share our enthusiasm for the sequent calculus and languages built on it with more people, and with this pearl, we hope that others will start to write their own little compilers to the sequent calculus and explore the exciting possibilities it offers.

%% file: sections/sequent-calculus.tex
In the main part of the paper we introduced the \lambdamumu\ without any references to the sequent calculus, because we think it is not essential to understand the latter in order to understand the former.
In this appendix, we provide the details which help make the connection between the logical calculus and the term system clear.
We only discuss a very simple sequent calculus which contains two logical connectives: the two conjunctions $A \otimes B$ and $A \with B$ which correspond to the strict and lazy pairs that we have seen in \targetlang.
We use $X$ for propositional variables.
\begin{equation*}
    A, B \Coloneqq X \mid A \otimes B \mid A \with B
\end{equation*}
In the (classical) sequent calculus both the premisses and the conclusion of a derivation rule consist of \emph{sequents} $\Gamma \vdash \Delta$.
Both $\Gamma$ and $\Delta$ are multisets of formulas; that is, it is important how often a formula occurs on the left or the right, but not in which order the formulas occur.
In the sequent calculus, we only have \emph{introduction rules}.
This means that the logically complex formula $A \otimes B$ or $A \with B$ only occurs in the conclusion of the rules that define it, and not in one of the premises.
Every connective comes with a set of rules which introduce the connective on the left and the right of the turnstile.
In our case, the rules look like this:

\begin{minipage}{0.45\textwidth}
    \begin{prooftree}
        \AxiomC{}
        \RightLabel{\textsc{Axiom}}
        \UnaryInfC{$A \vdash A$}
    \end{prooftree}
\end{minipage}
\begin{minipage}{0.45\textwidth}
    \begin{prooftree}
        \AxiomC{$\Gamma_1 \vdash \Delta_1, A$}
        \AxiomC{$A, \Gamma_2 \vdash \Delta_2$}
        \RightLabel{\textsc{Cut}}
        \BinaryInfC{$\Gamma_1, \Gamma_2 \vdash \Delta_1, \Delta_2$}
    \end{prooftree}
\end{minipage}
\vspace{0.5cm}

\begin{minipage}{0.45\textwidth}
    \begin{prooftree}
        \AxiomC{$\Gamma, A_1, A_2 \vdash \Delta$}
        \RightLabel{$\otimes$-L}
        \UnaryInfC{$\Gamma, A_1 \otimes A_2 \vdash \Delta$}
    \end{prooftree}
\end{minipage}
\begin{minipage}{0.45\textwidth}
    \begin{prooftree}
        \AxiomC{$\Gamma_1 \vdash A_1, \Delta_1$}
        \AxiomC{$\Gamma_2 \vdash A_2, \Delta_2$}
        \RightLabel{$\otimes$-R}
        \BinaryInfC{$\Gamma_1, \Gamma_2 \vdash A_1 \otimes A_2, \Delta_1, \Delta_2$}
    \end{prooftree}
\end{minipage}
\vspace{0.5cm}

\begin{minipage}{0.3\textwidth}
    \begin{prooftree}
        \AxiomC{$\Gamma, A_1 \vdash \Delta$}
        \RightLabel{$\with$-L$_1$}
        \UnaryInfC{$\Gamma, A_1 \with A_2 \vdash \Delta$}
    \end{prooftree}
\end{minipage}
\begin{minipage}{0.3\textwidth}
    \begin{prooftree}
        \AxiomC{$\Gamma, A_2 \vdash \Delta$}
        \RightLabel{$\with$-L$_2$}
        \UnaryInfC{$\Gamma, A_1 \with A_2 \vdash \Delta$}
    \end{prooftree}
\end{minipage}
\begin{minipage}{0.3\textwidth}
    \begin{prooftree}
        \AxiomC{$\Gamma \vdash A_1, \Delta$}
        \AxiomC{$\Gamma \vdash A_2, \Delta$}
        \RightLabel{$\with$-R$_2$}
        \BinaryInfC{$\Gamma \vdash A_1 \with A_2, \Delta$}
    \end{prooftree}
\end{minipage}
\vspace{0.2cm}

The rule \textsc{Cut} is the only rule which destroys the so-called \emph{subformula property}.
This property says that every formula which occurs anywhere in a derivation is a subformula of a formula occurring in the conclusion of the derivation.
Proof theorists therefore try to show that we can \emph{eliminate the cuts}; if every sequent which can be derived using the \textsc{Cut} rule can also be derived without using it, we say that the calculus enjoys the \emph{cut-elimination property}.
The Curry-Howard correspondence for the sequent calculus relates this cut-elimination procedure to the computations that we have seen in the paper.

The first step from the sequent calculus towards the \lambdamumu\ consists in marking at most one of the formulas in each of the sequents as \emph{active}.
We mark a formula as active by enclosing it in a pair of brackets.
This yields two versions of the rule \textsc{Axiom}, one where we mark the formula on the left and one where we mark the formula on the right.
If we want to translate every derivation using the original rules to a derivation in the new variant we also have to add special rules which \emph{activate} and \emph{deactivate} formulas both on the left and on the right.
This yields the following new set of rules:

\begin{minipage}{0.25\textwidth}
    \begin{prooftree}
        \AxiomC{}
        \RightLabel{\textsc{Axiom-L}}
        \UnaryInfC{$[A] \vdash A$}
    \end{prooftree}
\end{minipage}
\begin{minipage}{0.25\textwidth}
    \begin{prooftree}
        \AxiomC{}
        \RightLabel{\textsc{Axiom-R}}
        \UnaryInfC{$A \vdash [A]$}
    \end{prooftree}
\end{minipage}
\begin{minipage}{0.4\textwidth}
    \begin{prooftree}
        \AxiomC{$\Gamma_1 \vdash \Delta_1, [A]$}
        \AxiomC{$[A], \Gamma_2 \vdash \Delta_2$}
        \RightLabel{\textsc{Cut}}
        \BinaryInfC{$\Gamma_1, \Gamma_2 \vdash \Delta_1, \Delta_2$}
    \end{prooftree}
\end{minipage}
\vspace{0.5cm}

\begin{minipage}{0.45\textwidth}
    \begin{prooftree}
        \AxiomC{$\Gamma_, A_1, A_2 \vdash \Delta$}
        \RightLabel{$\otimes$-L}
        \UnaryInfC{$\Gamma_, [A_1 \otimes A_2] \vdash \Delta$}
    \end{prooftree}
\end{minipage}
\begin{minipage}{0.45\textwidth}
    \begin{prooftree}
        \AxiomC{$\Gamma_1 \vdash [A_1], \Delta_1$}
        \AxiomC{$\Gamma_2 \vdash [A_2], \Delta_2$}
        \RightLabel{$\otimes$-R}
        \BinaryInfC{$\Gamma_1, \Gamma_2 \vdash [A_1 \otimes A_2], \Delta_1, \Delta_2$}
    \end{prooftree}
\end{minipage}
\vspace{0.5cm}

\begin{minipage}{0.3\textwidth}
    \begin{prooftree}
        \AxiomC{$\Gamma, [A_1] \vdash \Delta$}
        \RightLabel{$\with$-L$_1$}
        \UnaryInfC{$\Gamma, [A_1 \with A_2] \vdash \Delta$}
    \end{prooftree}
\end{minipage}
\begin{minipage}{0.3\textwidth}
    \begin{prooftree}
        \AxiomC{$\Gamma, [A_2] \vdash \Delta$}
        \RightLabel{$\with$-L$_2$}
        \UnaryInfC{$\Gamma, [A_1 \with A_2] \vdash \Delta$}
    \end{prooftree}
\end{minipage}
\begin{minipage}{0.3\textwidth}
    \begin{prooftree}
        \AxiomC{$\Gamma \vdash A_1, \Delta$}
        \AxiomC{$\Gamma \vdash A_2, \Delta$}
        \RightLabel{$\with$-R$_2$}
        \BinaryInfC{$\Gamma \vdash [A_1 \with A_2], \Delta$}
    \end{prooftree}
\end{minipage}
\vspace{0.5cm}

\begin{minipage}{0.23\textwidth}
    \begin{prooftree}
        \AxiomC{$\Gamma, A\vdash \Delta$}
        \RightLabel{\textsc{Act-L}}
        \UnaryInfC{$\Gamma, [A] \vdash \Delta$}
    \end{prooftree}    
\end{minipage}
\begin{minipage}{0.23\textwidth}
    \begin{prooftree}
        \AxiomC{$\Gamma, [A]\vdash \Delta$}
        \RightLabel{\textsc{Deact-L}}
        \UnaryInfC{$\Gamma, A \vdash \Delta$}
    \end{prooftree}    
\end{minipage}
\begin{minipage}{0.23\textwidth}
    \begin{prooftree}
        \AxiomC{$\Gamma \vdash A, \Delta$}
        \RightLabel{\textsc{Act-R}}
        \UnaryInfC{$\Gamma \vdash [A], \Delta$}
    \end{prooftree}    
\end{minipage}
\begin{minipage}{0.23\textwidth}
    \begin{prooftree}
        \AxiomC{$\Gamma \vdash [A], \Delta$}
        \RightLabel{\textsc{Deact-R}}
        \UnaryInfC{$\Gamma \vdash A, \Delta$}
    \end{prooftree}    
\end{minipage}

\vspace{0.2cm}

We can now begin to assign terms to derivations in this calculus by associating every non-active formula $A$ in the \emph{left} side of the turnstile with a producer variable $x \prd A$, and every non-active formula $B$ on the \emph{right} side of the turnstile with a consumer variable $\alpha \cnt B$.
As discussed in \cref{subsec:related-work:term-assignment} we write both producer and consumer variables in a joint context $\Gamma$ on the left-hand side of typing rules.
We have to distinguish three different sequents, depending on whether a formula is active, and if so, on which side the active formula occurs.
If there is no active formula, then we assign a \emph{statement} to the sequent, if the formula on the right is active, we assign a \emph{producer}, and if the formula on the left is active, we assign a \emph{consumer}.
For most rules the correspondence is clear: The rule \textsc{Axiom-R} corresponds to the typing rule \textsc{Var}$_1$ (and \textsc{Axiom-L} to \textsc{Var}$_2$).
The rule \textsc{Tup} corresponds to the rule $\otimes$-R, and \textsc{Case-Pair} to $\otimes$-L.
The rules \textsc{Fst} and \textsc{Snd} correspond to the rules $\with$-L$_1$ and $\with$-L$_2$, and \textsc{Cocase-LPair} to $\with$-R.
The activation rules correspond to the rules $\mu$ and $\tilde\mu$, and deactivation can be expressed as a cut with a variable.

%% file: sections/surface-typing.tex
Given a term $t$, an environment $\Gamma$ and a program $P$, if $t$ has type $\tau$ in environment $\Gamma$ and program $P$, we write $\Gamma \xvdash{P} t:\tau$.
As $P$ is only used for typing calls to top-level definitions (rule \textsc{Call}), we usually leave it implicit in the typing rules.
To make sure programs $P$ are well-formed, we have additional checking rules for programs \textsc{$\emptyset$-ok} and \textsc{P-Ok}.
If a program is well-formed, we write $\vdash P\textsc{ Ok}$.\\
\begin{minipage}{\textwidth}
  \begin{minipage}{0.3\textwidth}
    \begin{prooftree}
      \AxiomC{$x \prd \tau \in \Gamma$}
      \RightLabel{\textsc{Var}}
      \UnaryInfC{$\Gamma \vdash x : \tau$}
    \end{prooftree}
  \end{minipage}
  \hfill
  \begin{minipage}{0.3\textwidth}
    \begin{prooftree}
      \AxiomC{\phantom{$x : \tau \in \Gamma$}}
      \RightLabel{\textsc{Lit}}
      \UnaryInfC{$\Gamma \vdash \natlit{n}:\tyint$}
    \end{prooftree}
  \end{minipage}
  \hfill
  \begin{minipage}{0.3\textwidth}
    \begin{prooftree}
      \AxiomC{$\Gamma \vdash t_1: \tyint \quad \Gamma \vdash t_2: \tyint$}
      \RightLabel{\textsc{Op}}
      \UnaryInfC{$\Gamma \vdash t_1\odot t_2 : \tyint$}
    \end{prooftree}
  \end{minipage}
  \hfill
  \vspace{1em}

  \begin{minipage}{0.55\textwidth}
    \begin{prooftree}
      \AxiomC{$\Gamma \vdash n : \tyint$}
      \AxiomC{$\Gamma \vdash t_1 : \tau$}
      \AxiomC{$\Gamma \vdash t_2 : \tau$}
      \RightLabel{\textsc{Ifz}}
      \TrinaryInfC{$\Gamma \vdash \ifzero{n}{t_1}{t_2} : \tau$}
    \end{prooftree}
  \end{minipage}
  \hfill
  \begin{minipage}{0.4\textwidth}
    \begin{prooftree}
      \AxiomC{$\Gamma \vdash t_1 : \tau_1$}
      \AxiomC{$\Gamma, x \prd \tau_1 \vdash t_2 : \tau_2$}
      \RightLabel{\textsc{Let}}
      \BinaryInfC{$\Gamma \vdash \letin{x}{t_1}{t_2} :\tau_2$}
    \end{prooftree}
  \end{minipage}
  \hfill
  \vspace{1em}

  \begin{minipage}{\textwidth}
    \begin{prooftree}
      \AxiomC{$\mathbf{def}\ f(\overline{x_i \prd \tau_i};\overline{\alpha_j\cnt\tau_j}) :\tau \in P$}
      \AxiomC{$\overline{\Gamma \vdash t_i : \tau_i}$}
      \AxiomC{$\overline{\Gamma \vdash \alpha_j \cnt \tau_j}$}
      \RightLabel{\textsc{Call}}
      \TrinaryInfC{$\Gamma \xvdash{P} f(\overline{t_i};\overline{\alpha_j}) : \tau$}
    \end{prooftree}
  \end{minipage}
  \hfill
  \vspace{1em}

  \begin{minipage}{\textwidth}
    \begin{prooftree}
      \AxiomC{$\Gamma \vdash t:\mathtt{List}(\tau^{\prime})$}
      \AxiomC{$\Gamma\vdash t_1:\tau$}
      \AxiomC{$\Gamma,y\prd\tau^{\prime},z\prd\mathtt{List}(\tau^{\prime})\vdash t_2:\tau$}
      \RightLabel{\textsc{Case-List}}
      \TrinaryInfC{$\Gamma \vdash \case{t}{\mathtt{Nil}\Rightarrow t_1,\mathtt{Cons}(y,z) \Rightarrow t_2 } : \tau$}
    \end{prooftree}
  \end{minipage}
  \hfill
  \vspace{1em}

  \begin{minipage}{0.55\textwidth}
    \begin{prooftree}
      \AxiomC{$\Gamma \vdash t_1 : \tau$}
      \AxiomC{$\Gamma \vdash t_2 : \mathtt{List}(\tau)$}
      \RightLabel{\textsc{Cons}}
      \BinaryInfC{$\Gamma \vdash \mathtt{Cons}(t_1,t_2):\mathtt{List}(\tau)$}
    \end{prooftree}
  \end{minipage}
  \hfill
  \begin{minipage}{0.4\textwidth}
    \begin{prooftree}
      \AxiomC{\phantom{$\Gamma \vdash t_1 : \tau$}}
      \RightLabel{\textsc{Nil}}
      \UnaryInfC{$\Gamma \vdash \mathtt{Nil}:\mathtt{List}(\tau)$}
    \end{prooftree}
  \end{minipage}
  \hfill
  \vspace{1em}

  \begin{minipage}{0.6\textwidth}
    \begin{prooftree}
      \AxiomC{$\Gamma \vdash t:\mathtt{Pair}(\tau_1,\tau_2) \quad \Gamma,x\prd\tau_1,y\prd\tau_2 \vdash t:\tau$}
      \RightLabel{\textsc{Case-Pair}}
      \UnaryInfC{$\Gamma\vdash \case{t}{\mathtt{Tup}(x,y) \Rightarrow t}:\tau$}
    \end{prooftree}
  \end{minipage}
  \hfill
  \begin{minipage}{0.35\textwidth}
    \begin{prooftree}
      \AxiomC{$\Gamma \vdash t_1 : \tau_1 \quad \Gamma \vdash t_2 : \tau_2$}
      \RightLabel{\textsc{Tup}}
      \UnaryInfC{$\Gamma \vdash \mathtt{Tup}(t_1,t_2) : \mathtt{Pair}(\tau_1,\tau_2)$}
    \end{prooftree}
  \end{minipage}
  \hfill
  \vspace{1em}

  \begin{minipage}{0.55\textwidth}
    \begin{prooftree}
      \AxiomC{$\Gamma\vdash t:\mathtt{Stream}(\tau)$}
      \RightLabel{\textsc{Hd}}
      \UnaryInfC{$\Gamma \vdash t.\mathtt{hd} : \tau$}
    \end{prooftree}
  \end{minipage}
  \hfill
  \begin{minipage}{0.4\textwidth}
    \begin{prooftree}
      \AxiomC{$\Gamma \vdash t:\mathtt{Stream}(\tau)$}
      \RightLabel{\textsc{Tl}}
      \UnaryInfC{$\Gamma \vdash t.\mathtt{tl} : \mathtt{Stream}(\tau)$}
    \end{prooftree}
  \end{minipage}
  \hfill
  \vspace{1em}

  \begin{minipage}{\textwidth}
    \begin{prooftree}
      \AxiomC{$\Gamma\vdash t_1:\tau$}
      \AxiomC{$\Gamma\vdash t_2:\mathtt{Stream}(\tau)$}
      \RightLabel{\textsc{Stream}}
      \BinaryInfC{$\Gamma \vdash \cocase{\mathtt{hd}\Rightarrow t_1, \mathtt{tl}\Rightarrow t_2} : \mathtt{Stream}(\tau)$}
    \end{prooftree}
  \end{minipage}
  \hfill
  \vspace{1em}

  \begin{minipage}{0.55\textwidth}
    \begin{prooftree}
      \AxiomC{$\Gamma\vdash t:\mathtt{LPair}(\tau_1,\tau_2)$}
      \RightLabel{\textsc{fst}}
      \UnaryInfC{$\Gamma \vdash t.\mathtt{fst} : \tau_1$}
    \end{prooftree}
  \end{minipage}
  \hfill
  \begin{minipage}{0.4\textwidth}
    \begin{prooftree}
      \AxiomC{$\Gamma \vdash t:\mathtt{LPair}(\tau_1,\tau_2)$}
      \RightLabel{\textsc{snd}}
      \UnaryInfC{$\Gamma \vdash t.\mathtt{snd} : \tau_2$}
    \end{prooftree}
  \end{minipage}
  \hfill
  \vspace{1em}

  \begin{minipage}{\textwidth}
    \begin{prooftree}
      \AxiomC{$\Gamma\vdash t_1:\tau_1$}
      \AxiomC{$\Gamma\vdash t_2:\tau_2$}
      \RightLabel{\textsc{LPair}}
      \BinaryInfC{$\Gamma \vdash \cocase{\mathtt{fst}\Rightarrow t_1, \mathtt{snd}\Rightarrow t_2} : \mathtt{LPair}(\tau_1,\tau_2)$}
    \end{prooftree}
  \end{minipage}
  \hfill
  \vspace{1em}
 
  \begin{minipage}{0.55\textwidth}
    \begin{prooftree}
      \AxiomC{$\Gamma \vdash t_1 : \tau_1 \rightarrow \tau_2 \quad \Gamma \vdash t_2 : \tau_1$}
      \RightLabel{\textsc{App}}
      \UnaryInfC{$\Gamma \vdash t_1\ t_2 : \tau_2$}
    \end{prooftree}
  \end{minipage}
  \hfill
  \begin{minipage}{0.4\textwidth}
    \begin{prooftree}
      \AxiomC{$\Gamma,x\prd\tau_1 \vdash t:\tau_2$}
      \RightLabel{\textsc{Lam}}
      \UnaryInfC{$\Gamma \vdash \lambda x.t : \tau_1 \rightarrow \tau_2$}
    \end{prooftree}
  \end{minipage}
  \hfill
  \vspace{1em}

  \begin{minipage}{0.55\textwidth}
    \begin{prooftree}
      \AxiomC{$\Gamma \vdash t : \tau$}
      \AxiomC{$\alpha \cnt \tau \in \Gamma$}
      \RightLabel{\textsc{Goto}}
      \BinaryInfC{$\Gamma \vdash \jump{t}{\alpha} : \tau'$}
    \end{prooftree}
  \end{minipage}
  \begin{minipage}{0.4\textwidth}
    \begin{prooftree}
      \AxiomC{$\Gamma, \alpha \cnt \tau \vdash t : \tau$}
      \RightLabel{\textsc{Label}}
      \UnaryInfC{$\Gamma \vdash \labelterm{\alpha}{t} : \tau$}
    \end{prooftree}
  \end{minipage}
\end{minipage}
\hfill\\

To check a program, we start with the empty program, which we know is well-formed (\textsc{Wf-Empty}), and then add one definition at a time. 
A definition is then well-typed if there are types $\overline{\tau_i}$ and $\overline{\tau_j}$ for its arguments such that its body is well-typed.
Because we explicitly allow recursive definitions, the body $t$ might contain the name $\text{f}$ as well.
Thus, while we typecheck $t$, we add the definition of $\text{f}$ to the program and assume it is well-typed.
After finding $\overline{\tau_i}, \overline{\tau_j}$ and $\tau$, these are added to the program as well, that is, well-formed programs contain type annotations while definitions on their own do not.
This way, these types can be used while checking types of calls (in rule \textsc{Call}).
\\
\vspace{1em}
\begin{minipage}{0.25\textwidth}
  \begin{prooftree}
    \AxiomC{}
    \RightLabel{\textsc{Wf-Empty}}
    \UnaryInfC{$\vdash \wellformed{\emptyset}$}
  \end{prooftree}
\end{minipage}
\hfill
\begin{minipage}{0.7\textwidth}
  \begin{prooftree}
    \AxiomC{$\vdash \wellformed{P}$}
    \AxiomC{$\overline{x\prd\tau_i},\overline{\alpha\cnt\tau_j} \xvdash{P,\mathbf{def}\ \text{f}(\overline{x_i:\tau_i},\overline{\alpha_j\cnt\tau_j}):\tau\coloneq t} t : \tau$}
    \RightLabel{\textsc{Wf-Cons}}
    \BinaryInfC{$\vdash \wellformed{P,\mathbf{def}\ \text{f}(\overline{x_i:\tau_i},\overline{\alpha_j\cnt\tau_j}) : \tau \coloneq t}$}
  \end{prooftree}
\end{minipage}
\hfill

%% file: sections/semantics-label.tex
The full operational semantics for the $\mathbf{label}/\mathbf{goto}$ construct is in essence the same as for $\mathtt{let/cc}$.
To make it precise, we promote evaluation contexts to runtime values.

We first repeat \cref{def:focusing:evaluationcontexts} of evaluation contexts with one change: $\labelterm{\alpha}{E}$ is not an evaluation context. We reduce a $\mathbf{label}$ as soon as it comes into evaluation position.
\[
  \begin{array}{rcl}
    E & \Coloneqq & \square \mid E\odot t \mid \valueof{t} \odot E \mid \ifzero{E}{t}{t} \mid \letin{x}{E}{t} \mid f(\overline{\valueof{t}},E,\overline{t})\mid K(\overline{\valueof{t}},E,\overline{t})\\
      & \mid & \case{E}{\overline{K(\overline{x}) \Rightarrow t}}\mid E\ t \mid \valueof{t}\ E \mid E.D(\overline{t}) \mid \valueof{t}.D(\overline{\valueof{t}},E,\overline{t}) \mid \jump{E}{\alpha}
  \end{array}
\]
Now we add them as another form of value
\[
  \valueof{t} \Coloneqq \ldots \mid E
\]
Note that these values only exist at runtime, that is, they cannot appear in expressions before evaluation has started.
They are typed as consumers, which means that they are the only values with a consumer type.
This makes sure that we can substitute them for covariables.
Their typing can be captured by the following rule.
\begin{prooftree}
  \AxiomC{$x \prd \tau \vdash E[x] : \tau_0$}
  \RightLabel{\textsc{Ctx}}
  \UnaryInfC{$\vdash E \cnt \tau$}
\end{prooftree}
The rule means that if the hole of a context $E$ expects an expression of type $\tau$ to be plugged in, then we have $E \cnt \tau$.

Now we can give the evaluation rules for $\mathbf{label}$ and $\mathbf{goto}$:
\begin{equation*}
  E[\labelterm{\alpha}{t}] \reducesto E[t[E/\alpha]] \qquad\qquad
  E^{\prime}[\jump{\valueof{t}}{E}] \reducesto E[\valueof{t}]
\end{equation*}
In the rule for $\mathbf{label}$, the surrounding evaluation context $E$ is reified as a value and then substituted for the covariable $\alpha$ in the body $t$.
Note that $E$ is not removed, i.e., evaluation continues in this context.
In particular, if $\alpha$ does not occur free in $t$, then the $\mathbf{label}$ is effectively a no-op.
We can also see that the types are correct: If $t$ has type $\tau$, then so does $\labelterm{\alpha}{t}$ and consequently we have $E \cnt \tau$ which is the same type as that of $\alpha$.
In the rule for $\mathbf{goto}$ the covariable must have already been replaced by an evaluation context, which is ensured if the evaluated term was closed and well-typed, because the only way to introduce a covariable is through a $\mathbf{label}$.
The evaluation step then removes and discards the surrounding context $E^{\prime}$ and continues evaluation by plugging the value $\valueof{t}$ into the previously reified context $E$.
Note that $E^{\prime}$ cannot contain $\mathbf{label}$s, as they are not evaluation contexts.
This ensures that there is no risk of removing a binder for a free variable in $\valueof{t}$.
Together these two rules also allow us to simulate the approximate rule from \cref{subsec:focusing:fun}:
\begin{equation*}
  E[\labelterm{\alpha}{E^{\prime}[\jump{\valueof{t}}{\alpha}]}]
  \reducesto E[E^{\prime}[\jump{\valueof{t}}{E}]]
  \reducesto E[\valueof{t}]
\end{equation*}

The rule for $\mathbf{goto}$ also is the reason why the theorem of strong preservation does not immediately hold (see the discussion in \cref{subsec:typing:theorems}).
The problem is that from this rule and the given typing rules for \surfacelang\ it is not immediate that the evaluation contexts $E^{\prime}$ and $E$ yield a term of the same type when filling their holes, so that the overall type of the term may not be preserved.
But this cannot actually happen, because all other reduction rules preserve the overall type and hence all evaluation contexts that are reified by the rule for $\mathbf{label}$ must yield a term of that same overall type when their holes are filled.
Therefore, also the rule for $\mathbf{goto}$ is type-preserving.
This can be made precise by explicitly tracking the overall type in the type system (see, e.g., Section 6 in \cite{Wright1994Soundness}).

%% file: sc-intro.bbl

\begin{thebibliography}{44}


\ifx \showCODEN    \undefined \def \showCODEN     #1{\unskip}     \fi
\ifx \showDOI      \undefined \def \showDOI       #1{#1}\fi
\ifx \showISBNx    \undefined \def \showISBNx     #1{\unskip}     \fi
\ifx \showISBNxiii \undefined \def \showISBNxiii  #1{\unskip}     \fi
\ifx \showISSN     \undefined \def \showISSN      #1{\unskip}     \fi
\ifx \showLCCN     \undefined \def \showLCCN      #1{\unskip}     \fi
\ifx \shownote     \undefined \def \shownote      #1{#1}          \fi
\ifx \showarticletitle \undefined \def \showarticletitle #1{#1}   \fi
\ifx \showURL      \undefined \def \showURL       {\relax}        \fi
\providecommand\bibfield[2]{#2}
\providecommand\bibinfo[2]{#2}
\providecommand\natexlab[1]{#1}
\providecommand\showeprint[2][]{arXiv:#2}

\bibitem[Abel et~al\mbox{.}(2013)]%
        {Abel2013copatterns}
\bibfield{author}{\bibinfo{person}{Andreas Abel}, \bibinfo{person}{Brigitte
  Pientka}, \bibinfo{person}{David Thibodeau}, {and} \bibinfo{person}{Anton
  Setzer}.} \bibinfo{year}{2013}\natexlab{}.
\newblock \showarticletitle{Copatterns: Programming Infinite Structures by
  Observations}. In \bibinfo{booktitle}{\emph{Proceedings of the 40th Annual
  ACM SIGPLAN-SIGACT Symposium on Principles of Programming Languages}} (Rome,
  Italy) \emph{(\bibinfo{series}{POPL '13})}. \bibinfo{publisher}{Association
  for Computing Machinery}, \bibinfo{address}{New York, NY, USA},
  \bibinfo{pages}{27--38}.
\newblock
\showISBNx{9781450318327}
\urldef\tempurl%
\url{https://doi.org/10.1145/2480359.2429075}
\showURL{%
\tempurl}


\bibitem[Andreoli(1992)]%
        {Andreoli1992logicprogramming}
\bibfield{author}{\bibinfo{person}{Jean-Marc Andreoli}.}
  \bibinfo{year}{1992}\natexlab{}.
\newblock \showarticletitle{Logic Programming with Focusing Proofs in Linear
  Logic}.
\newblock \bibinfo{journal}{\emph{Journal of Logic and Computation}}
  \bibinfo{volume}{2} (\bibinfo{year}{1992}), \bibinfo{pages}{297--347}.
\newblock
Issue 3.
\urldef\tempurl%
\url{https://doi.org/10.1093/logcom/2.3.297}
\showURL{%
\tempurl}


\bibitem[Brachth\"{a}user et~al\mbox{.}(2020)]%
        {Brachthaeuser2020capabilities}
\bibfield{author}{\bibinfo{person}{Jonathan~Immanuel Brachth\"{a}user},
  \bibinfo{person}{Philipp Schuster}, {and} \bibinfo{person}{Klaus Ostermann}.}
  \bibinfo{year}{2020}\natexlab{}.
\newblock \showarticletitle{Effects as capabilities: effect handlers and
  lightweight effect polymorphism}.
\newblock \bibinfo{journal}{\emph{Proc. ACM Program. Lang.}}
  \bibinfo{volume}{4}, \bibinfo{number}{OOPSLA}, Article
  \bibinfo{articleno}{126} (\bibinfo{date}{nov} \bibinfo{year}{2020}),
  \bibinfo{numpages}{30}~pages.
\newblock
\urldef\tempurl%
\url{https://doi.org/10.1145/3428194}
\showDOI{\tempurl}


\bibitem[Cook(2009)]%
        {Cook2009understanding}
\bibfield{author}{\bibinfo{person}{William~R. Cook}.}
  \bibinfo{year}{2009}\natexlab{}.
\newblock \showarticletitle{On Understanding Data Abstraction, Revisited}. In
  \bibinfo{booktitle}{\emph{Proceedings of the Conference on Object-Oriented
  Programming, Systems, Languages and Applications: Onward! Essays}} (Orlando).
  \bibinfo{publisher}{Association for Computing Machinery},
  \bibinfo{address}{New York, NY, USA}, \bibinfo{pages}{557--572}.
\newblock
\urldef\tempurl%
\url{https://doi.org/10.1145/1640089.1640133}
\showDOI{\tempurl}


\bibitem[Curien and Herbelin(2000)]%
        {Curien2000duality}
\bibfield{author}{\bibinfo{person}{Pierre-Louis Curien} {and}
  \bibinfo{person}{Hugo Herbelin}.} \bibinfo{year}{2000}\natexlab{}.
\newblock \showarticletitle{The Duality of Computation}. In
  \bibinfo{booktitle}{\emph{Proceedings of the Fifth ACM SIGPLAN International
  Conference on Functional Programming}} \emph{(\bibinfo{series}{ICFP '00})}.
  \bibinfo{publisher}{Association for Computing Machinery},
  \bibinfo{address}{New York, NY, USA}, \bibinfo{pages}{233--243}.
\newblock
\urldef\tempurl%
\url{https://doi.org/10.1145/357766.351262}
\showURL{%
\tempurl}


\bibitem[Curien and Munch-Maccagnoni(2010)]%
        {Curien2010}
\bibfield{author}{\bibinfo{person}{Pierre-Louis Curien} {and}
  \bibinfo{person}{Guillaume Munch-Maccagnoni}.}
  \bibinfo{year}{2010}\natexlab{}.
\newblock \showarticletitle{The Duality of Computation under Focus}. In
  \bibinfo{booktitle}{\emph{Theoretical Computer Science}},
  \bibfield{editor}{\bibinfo{person}{Cristian~S. Calude} {and}
  \bibinfo{person}{Vladimiro Sassone}} (Eds.). \bibinfo{publisher}{Springer
  Berlin Heidelberg}, \bibinfo{address}{Berlin, Heidelberg},
  \bibinfo{pages}{165--181}.
\newblock
\showISBNx{978-3-642-15240-5}


\bibitem[Downen and Ariola(2014)]%
        {Downen2014duality}
\bibfield{author}{\bibinfo{person}{Paul Downen} {and} \bibinfo{person}{Zena~M.
  Ariola}.} \bibinfo{year}{2014}\natexlab{}.
\newblock \showarticletitle{The Duality of Construction}. In
  \bibinfo{booktitle}{\emph{Proceedings of the 23rd European Symposium on
  Programming Languages and Systems - Volume 8410}}
  \emph{(\bibinfo{series}{ESOP '14})}. \bibinfo{publisher}{Springer},
  \bibinfo{address}{Berlin, Heidelberg}, \bibinfo{pages}{249--269}.
\newblock
\urldef\tempurl%
\url{https://doi.org/10.1007/978-3-642-54833-8_14}
\showURL{%
\tempurl}


\bibitem[Downen and Ariola(2018a)]%
        {Downen2018beyondPolarity}
\bibfield{author}{\bibinfo{person}{Paul Downen} {and} \bibinfo{person}{Zena~M.
  Ariola}.} \bibinfo{year}{2018}\natexlab{a}.
\newblock \showarticletitle{{Beyond Polarity: Towards a Multi-Discipline
  Intermediate Language with Sharing}}. In \bibinfo{booktitle}{\emph{27th EACSL
  Annual Conference on Computer Science Logic (CSL 2018)}}
  \emph{(\bibinfo{series}{Leibniz International Proceedings in Informatics
  (LIPIcs)}, Vol.~\bibinfo{volume}{119})},
  \bibfield{editor}{\bibinfo{person}{Dan Ghica} {and} \bibinfo{person}{Achim
  Jung}} (Eds.). \bibinfo{publisher}{Schloss Dagstuhl - Leibniz-Zentrum
  f{\"{u}}r Informatik}, \bibinfo{address}{Dagstuhl, Germany},
  \bibinfo{pages}{21:1--21:23}.
\newblock
\showISBNx{978-3-95977-088-0}
\showISSN{1868-8969}
\urldef\tempurl%
\url{https://doi.org/10.4230/LIPIcs.CSL.2018.21}
\showDOI{\tempurl}


\bibitem[Downen and Ariola(2018b)]%
        {Downen2018}
\bibfield{author}{\bibinfo{person}{Paul Downen} {and} \bibinfo{person}{Zena~M.
  Ariola}.} \bibinfo{year}{2018}\natexlab{b}.
\newblock \showarticletitle{A tutorial on computational classical logic and the
  sequent calculus}.
\newblock \bibinfo{journal}{\emph{Journal of Functional Programming}}
  \bibinfo{volume}{28} (\bibinfo{year}{2018}).
\newblock
\urldef\tempurl%
\url{https://doi.org/10.1017/S0956796818000023}
\showURL{%
\tempurl}


\bibitem[Downen and Ariola(2020)]%
        {Downen2020compiling}
\bibfield{author}{\bibinfo{person}{Paul Downen} {and} \bibinfo{person}{Zena~M.
  Ariola}.} \bibinfo{year}{2020}\natexlab{}.
\newblock \showarticletitle{{Compiling With Classical Connectives}}.
\newblock \bibinfo{journal}{\emph{{Logical Methods in Computer Science}}}
  \bibinfo{volume}{{Volume 16, Issue 3}} (\bibinfo{date}{Aug.}
  \bibinfo{year}{2020}).
\newblock
\urldef\tempurl%
\url{https://doi.org/10.23638/LMCS-16(3:13)2020}
\showDOI{\tempurl}


\bibitem[Downen et~al\mbox{.}(2015)]%
        {Downen2015structural}
\bibfield{author}{\bibinfo{person}{Paul Downen}, \bibinfo{person}{Philip
  Johnson-Freyd}, {and} \bibinfo{person}{Zena~M. Ariola}.}
  \bibinfo{year}{2015}\natexlab{}.
\newblock \showarticletitle{Structures for structural recursion}. In
  \bibinfo{booktitle}{\emph{Proceedings of the 20th ACM SIGPLAN International
  Conference on Functional Programming}} (Vancouver, BC, Canada)
  \emph{(\bibinfo{series}{ICFP 2015})}. \bibinfo{publisher}{Association for
  Computing Machinery}, \bibinfo{address}{New York, NY, USA},
  \bibinfo{pages}{127–139}.
\newblock
\showISBNx{9781450336697}
\urldef\tempurl%
\url{https://doi.org/10.1145/2784731.2784762}
\showDOI{\tempurl}


\bibitem[Downen et~al\mbox{.}(2016)]%
        {Downen2016sequent}
\bibfield{author}{\bibinfo{person}{Paul Downen}, \bibinfo{person}{Luke Maurer},
  \bibinfo{person}{Zena~M Ariola}, {and} \bibinfo{person}{Simon Peyton~Jones}.}
  \bibinfo{year}{2016}\natexlab{}.
\newblock \showarticletitle{Sequent calculus as a compiler intermediate
  language}. In \bibinfo{booktitle}{\emph{Proceedings of the 21st ACM SIGPLAN
  International Conference on Functional Programming}}.
  \bibinfo{pages}{74--88}.
\newblock


\bibitem[Downen et~al\mbox{.}(2019)]%
        {Downen2019codata}
\bibfield{author}{\bibinfo{person}{Paul Downen}, \bibinfo{person}{Zachary
  Sullivan}, \bibinfo{person}{Zena~M. Ariola}, {and}
  \bibinfo{person}{Simon~Peyton Jones}.} \bibinfo{year}{2019}\natexlab{}.
\newblock \showarticletitle{Codata in Action}. In
  \bibinfo{booktitle}{\emph{European Symposium on Programming}}
  \emph{(\bibinfo{series}{ESOP '19})}. Springer, \bibinfo{pages}{119--146}.
\newblock
\urldef\tempurl%
\url{https://doi.org/10.1007/978-3-030-17184-1_5}
\showURL{%
\tempurl}


\bibitem[Felleisen(1987)]%
        {Felleisen1987}
\bibfield{author}{\bibinfo{person}{Matthias Felleisen}.}
  \bibinfo{year}{1987}\natexlab{}.
\newblock \showarticletitle{Reflections on Landin's J-operator: A Partly
  Historical Note}.
\newblock \bibinfo{journal}{\emph{Computer Languages}} \bibinfo{volume}{12},
  \bibinfo{number}{3} (\bibinfo{year}{1987}), \bibinfo{pages}{197--207}.
\newblock
\showISSN{0096-0551}
\urldef\tempurl%
\url{https://doi.org/10.1016/0096-0551(87)90022-1}
\showDOI{\tempurl}


\bibitem[Felleisen et~al\mbox{.}(1987)]%
        {Felleisen1987syntactic}
\bibfield{author}{\bibinfo{person}{Matthias Felleisen},
  \bibinfo{person}{Daniel~P. Friedman}, \bibinfo{person}{Eugene Kohlbecker},
  {and} \bibinfo{person}{Bruce Duba}.} \bibinfo{year}{1987}\natexlab{}.
\newblock \showarticletitle{A syntactic theory of sequential control}.
\newblock \bibinfo{journal}{\emph{Theoretical Computer Science}}
  \bibinfo{volume}{52}, \bibinfo{number}{3} (\bibinfo{year}{1987}),
  \bibinfo{pages}{205--237}.
\newblock
\showISSN{0304-3975}
\urldef\tempurl%
\url{https://doi.org/10.1016/0304-3975(87)90109-5}
\showDOI{\tempurl}


\bibitem[Filinski(1989)]%
        {Filinski1989}
\bibfield{author}{\bibinfo{person}{Andrzej Filinski}.}
  \bibinfo{year}{1989}\natexlab{}.
\newblock \showarticletitle{Declarative Continuations: an Investigation of
  Duality in Programming Language Semantics}. In
  \bibinfo{booktitle}{\emph{Category Theory and Computer Science}}.
  \bibinfo{publisher}{Springer-Verlag}, \bibinfo{address}{Berlin, Heidelberg},
  \bibinfo{pages}{224–249}.
\newblock
\showISBNx{354051662X}


\bibitem[Gentzen(1935a)]%
        {Gentzen1935a}
\bibfield{author}{\bibinfo{person}{Gerhard Gentzen}.}
  \bibinfo{year}{1935}\natexlab{a}.
\newblock \showarticletitle{Untersuchungen über das logische Schließen. I.}
\newblock \bibinfo{journal}{\emph{Mathematische Zeitschrift}}
  \bibinfo{volume}{35} (\bibinfo{year}{1935}), \bibinfo{pages}{176--210}.
\newblock


\bibitem[Gentzen(1935b)]%
        {Gentzen1935b}
\bibfield{author}{\bibinfo{person}{Gerhard Gentzen}.}
  \bibinfo{year}{1935}\natexlab{b}.
\newblock \showarticletitle{Untersuchungen über das logische Schließen. II.}
\newblock \bibinfo{journal}{\emph{Mathematische Zeitschrift}}
  \bibinfo{volume}{39} (\bibinfo{year}{1935}), \bibinfo{pages}{405--431}.
\newblock


\bibitem[Gentzen(1969)]%
        {Gentzen1969}
\bibfield{author}{\bibinfo{person}{Gerhard Gentzen}.}
  \bibinfo{year}{1969}\natexlab{}.
\newblock \bibinfo{booktitle}{\emph{The collected papers of Gerhard Gentzen}}.
\newblock \bibinfo{publisher}{North-Holland Publishing Co.},
  \bibinfo{address}{Amsterdam}.
\newblock


\bibitem[Girard(1987)]%
        {Girard1987}
\bibfield{author}{\bibinfo{person}{Jean-Yves Girard}.}
  \bibinfo{year}{1987}\natexlab{}.
\newblock \showarticletitle{Linear Logic}.
\newblock \bibinfo{journal}{\emph{Theoretical Computer Science}}
  \bibinfo{volume}{50}, \bibinfo{number}{1} (\bibinfo{year}{1987}),
  \bibinfo{pages}{1--101}.
\newblock
\showISSN{0304-3975}
\urldef\tempurl%
\url{https://doi.org/10.1016/0304-3975(87)90045-4}
\showURL{%
\tempurl}


\bibitem[Goetz et~al\mbox{.}(2014)]%
        {Goetz2014jsr}
\bibfield{author}{\bibinfo{person}{Brian Goetz} {et~al\mbox{.}}}
  \bibinfo{year}{2014}\natexlab{}.
\newblock \bibinfo{booktitle}{\emph{JSR 335: Lambda Expressions for the Java
  Programming Language}}.
\newblock
\urldef\tempurl%
\url{https://jcp.org/en/jsr/detail?id=335}
\showURL{%
\tempurl}


\bibitem[Griffin(1989)]%
        {Griffin1989formulae}
\bibfield{author}{\bibinfo{person}{Timothy~G. Griffin}.}
  \bibinfo{year}{1989}\natexlab{}.
\newblock \showarticletitle{A Formulae-as-Type Notion of Control}. In
  \bibinfo{booktitle}{\emph{{Proceedings of the 17th ACM SIGPLAN-SIGACT
  Symposium on Principles of Programming Languages}}} (San Francisco,
  California, USA) \emph{(\bibinfo{series}{POPL '90})}.
  \bibinfo{publisher}{Association for Computing Machinery},
  \bibinfo{address}{New York, NY, USA}, \bibinfo{pages}{47--58}.
\newblock
\urldef\tempurl%
\url{https://doi.org/10.1145/96709.96714}
\showDOI{\tempurl}


\bibitem[Hagino(1989)]%
        {Hagino1989codatatypes}
\bibfield{author}{\bibinfo{person}{Tatsuya Hagino}.}
  \bibinfo{year}{1989}\natexlab{}.
\newblock \showarticletitle{Codatatypes in ML}.
\newblock \bibinfo{journal}{\emph{Journal of Symbolic Computation}}
  \bibinfo{volume}{8}, \bibinfo{number}{6} (\bibinfo{year}{1989}),
  \bibinfo{pages}{629--650}.
\newblock
\urldef\tempurl%
\url{https://doi.org/10.1016/S0747-7171(89)80065-3}
\showURL{%
\tempurl}


\bibitem[Landin(1965)]%
        {Landin1965}
\bibfield{author}{\bibinfo{person}{Peter~John Landin}.}
  \bibinfo{year}{1965}\natexlab{}.
\newblock \showarticletitle{Correspondence between ALGOL 60 and Church's
  Lambda-notation: part I}.
\newblock \bibinfo{journal}{\emph{Commun. ACM}} \bibinfo{volume}{8},
  \bibinfo{number}{2} (\bibinfo{date}{feb} \bibinfo{year}{1965}),
  \bibinfo{pages}{89--101}.
\newblock
\showISSN{0001-0782}
\urldef\tempurl%
\url{https://doi.org/10.1145/363744.363749}
\showDOI{\tempurl}


\bibitem[Levy(1999)]%
        {Levy1999}
\bibfield{author}{\bibinfo{person}{Paul~Blain Levy}.}
  \bibinfo{year}{1999}\natexlab{}.
\newblock \showarticletitle{Call-by-Push-Value: A Subsuming Paradigm}. In
  \bibinfo{booktitle}{\emph{Proceedings of the 4th International Conference on
  Typed Lambda Calculi and Applications}} \emph{(\bibinfo{series}{TLCA '99})}.
  \bibinfo{publisher}{Springer-Verlag}, \bibinfo{address}{Berlin, Heidelberg},
  \bibinfo{pages}{228–242}.
\newblock
\showISBNx{3540657630}


\bibitem[Maurer et~al\mbox{.}(2017)]%
        {Maurer2017}
\bibfield{author}{\bibinfo{person}{Luke Maurer}, \bibinfo{person}{Paul Downen},
  \bibinfo{person}{Zena~M. Ariola}, {and} \bibinfo{person}{Simon
  Peyton~Jones}.} \bibinfo{year}{2017}\natexlab{}.
\newblock \showarticletitle{Compiling without Continuations}. In
  \bibinfo{booktitle}{\emph{Proceedings of the 38th ACM SIGPLAN Conference on
  Programming Language Design and Implementation}} (Barcelona, Spain)
  \emph{(\bibinfo{series}{PLDI 2017})}. \bibinfo{publisher}{Association for
  Computing Machinery}, \bibinfo{address}{New York, NY, USA},
  \bibinfo{pages}{482--494}.
\newblock
\showISBNx{9781450349888}
\urldef\tempurl%
\url{https://doi.org/10.1145/3062341.3062380}
\showDOI{\tempurl}


\bibitem[Miquey(2019)]%
        {Miquey2019}
\bibfield{author}{\bibinfo{person}{\'{E}tienne Miquey}.}
  \bibinfo{year}{2019}\natexlab{}.
\newblock \showarticletitle{A Classical Sequent Calculus with Dependent Types}.
\newblock \bibinfo{journal}{\emph{ACM Trans. Program. Lang. Syst.}}
  \bibinfo{volume}{41}, \bibinfo{number}{2}, Article \bibinfo{articleno}{8}
  (\bibinfo{date}{mar} \bibinfo{year}{2019}), \bibinfo{numpages}{47}~pages.
\newblock
\showISSN{0164-0925}
\urldef\tempurl%
\url{https://doi.org/10.1145/3230625}
\showDOI{\tempurl}


\bibitem[Munch-Maccagnoni(2009)]%
        {MunchMaccagnoni2009}
\bibfield{author}{\bibinfo{person}{Guillaume Munch-Maccagnoni}.}
  \bibinfo{year}{2009}\natexlab{}.
\newblock \showarticletitle{Focalisation and Classical Realisability}. In
  \bibinfo{booktitle}{\emph{Computer Science Logic: 23rd international
  Workshop, {CSL} 2009, 18th Annual Conference of the {EACSL}}} (Coimbra,
  Portugal) \emph{(\bibinfo{series}{CSL '09})},
  \bibfield{editor}{\bibinfo{person}{Erich Gr{\"a}del} {and}
  \bibinfo{person}{Reinhard Kahle}} (Eds.). \bibinfo{publisher}{Springer},
  \bibinfo{address}{Berlin, Heidelberg}, \bibinfo{pages}{409--423}.
\newblock
\urldef\tempurl%
\url{https://doi.org/10.1007/978-3-642-04027-6_30}
\showURL{%
\tempurl}


\bibitem[Munch-Maccagnoni(2013)]%
        {Munch2013phd}
\bibfield{author}{\bibinfo{person}{Guillaume Munch-Maccagnoni}.}
  \bibinfo{year}{2013}\natexlab{}.
\newblock \emph{\bibinfo{title}{{S}yntax and {M}odels of a non-{A}ssociative
  {C}omposition of {P}rograms and {P}roofs}}.
\newblock \bibinfo{thesistype}{Ph.\,D. Dissertation}. \bibinfo{school}{Univ.
  Paris Diderot}.
\newblock


\bibitem[Negri and Von~Plato(2001)]%
        {Negri2001structural}
\bibfield{author}{\bibinfo{person}{Sara Negri} {and} \bibinfo{person}{Jan
  Von~Plato}.} \bibinfo{year}{2001}\natexlab{}.
\newblock \bibinfo{booktitle}{\emph{Structural Proof Theory}}.
\newblock \bibinfo{publisher}{Cambridge University Press}.
\newblock
\urldef\tempurl%
\url{https://doi.org/10.1017/CBO9780511527340}
\showURL{%
\tempurl}


\bibitem[Ostermann et~al\mbox{.}(2022)]%
        {icfp2022}
\bibfield{author}{\bibinfo{person}{Klaus Ostermann}, \bibinfo{person}{David
  Binder}, \bibinfo{person}{Ingo Skupin}, \bibinfo{person}{Tim
  S\"{u}berkr\"{u}b}, {and} \bibinfo{person}{Paul Downen}.}
  \bibinfo{year}{2022}\natexlab{}.
\newblock \showarticletitle{Introduction and Elimination, Left and Right}.
\newblock \bibinfo{journal}{\emph{Proc. ACM Program. Lang.}}
  \bibinfo{volume}{6}, \bibinfo{number}{ICFP}, Article \bibinfo{articleno}{106}
  (\bibinfo{year}{2022}), \bibinfo{numpages}{28}~pages.
\newblock
\urldef\tempurl%
\url{https://doi.org/10.1145/3547637}
\showDOI{\tempurl}


\bibitem[Parigot(1992)]%
        {Parigot1992}
\bibfield{author}{\bibinfo{person}{Michel Parigot}.}
  \bibinfo{year}{1992}\natexlab{}.
\newblock \showarticletitle{$\lambda$$\mu$-Calculus: An algorithmic
  interpretation of classical natural deduction}. In
  \bibinfo{booktitle}{\emph{Logic Programming and Automated Reasoning}},
  \bibfield{editor}{\bibinfo{person}{Andrei Voronkov}} (Ed.).
  \bibinfo{publisher}{Springer}, \bibinfo{address}{Berlin, Heidelberg},
  \bibinfo{pages}{190--201}.
\newblock


\bibitem[Reynolds(1972)]%
        {Reynolds1972definitional}
\bibfield{author}{\bibinfo{person}{John~Charles Reynolds}.}
  \bibinfo{year}{1972}\natexlab{}.
\newblock \showarticletitle{Definitional Interpreters for Higher-Order
  Programming Languages}. In \bibinfo{booktitle}{\emph{ACMConf}} (Boston).
  \bibinfo{publisher}{Association for Computing Machinery},
  \bibinfo{address}{New York, NY, USA}, \bibinfo{pages}{717--740}.
\newblock
\urldef\tempurl%
\url{https://doi.org/10.1145/800194.805852}
\showURL{%
\tempurl}


\bibitem[Spiwack(2014)]%
        {Spiwack2014}
\bibfield{author}{\bibinfo{person}{Arnaud Spiwack}.}
  \bibinfo{year}{2014}\natexlab{}.
\newblock \bibinfo{title}{A Dissection of L}.  (\bibinfo{year}{2014}).
\newblock
\newblock
\shownote{Unpublished draft}.


\bibitem[Sørensen and Urzyczyn(2006)]%
        {SorensenUrzyczyn2006}
\bibfield{author}{\bibinfo{person}{Morten~Heine Sørensen} {and}
  \bibinfo{person}{Pawe{\l} Urzyczyn}.} \bibinfo{year}{2006}\natexlab{}.
\newblock \bibinfo{booktitle}{\emph{{Lectures on the Curry-Howard
  Isomorphism}}}. \bibinfo{series}{Studies in Logic and the Foundations of
  Mathematics}, Vol.~\bibinfo{volume}{149}.
\newblock \bibinfo{publisher}{Elsevier}.
\newblock


\bibitem[Thielecke(1998)]%
        {Thielecke1998}
\bibfield{author}{\bibinfo{person}{Hayo Thielecke}.}
  \bibinfo{year}{1998}\natexlab{}.
\newblock \showarticletitle{An Introduction to Landin‘s “A Generalization
  of Jumps and Labels”}.
\newblock \bibinfo{journal}{\emph{Higher Order Symbol. Comput.}}
  \bibinfo{volume}{11}, \bibinfo{number}{2} (\bibinfo{date}{sep}
  \bibinfo{year}{1998}), \bibinfo{pages}{117–123}.
\newblock
\showISSN{1388-3690}
\urldef\tempurl%
\url{https://doi.org/10.1023/A:1010060315625}
\showDOI{\tempurl}


\bibitem[Troelstra and Schwichtenberg(2000)]%
        {TroelstraSchwichtenberg2000}
\bibfield{author}{\bibinfo{person}{Anne~Sjerp Troelstra} {and}
  \bibinfo{person}{Helmut Schwichtenberg}.} \bibinfo{year}{2000}\natexlab{}.
\newblock \bibinfo{booktitle}{\emph{Basic Proof Theory, Second Edition}}.
\newblock \bibinfo{publisher}{Cambridge University Press}.
\newblock


\bibitem[Wadler(1990)]%
        {Wadler1990lineartypes}
\bibfield{author}{\bibinfo{person}{Philip Wadler}.}
  \bibinfo{year}{1990}\natexlab{}.
\newblock \showarticletitle{Linear Types Can Change the World!}. In
  \bibinfo{booktitle}{\emph{Programming Concepts and Methods}}.
  \bibinfo{publisher}{North-Holland}.
\newblock


\bibitem[Wadler(2003)]%
        {Wadler2003call}
\bibfield{author}{\bibinfo{person}{Philip Wadler}.}
  \bibinfo{year}{2003}\natexlab{}.
\newblock \showarticletitle{Call-by-value is dual to call-by-name}. In
  \bibinfo{booktitle}{\emph{Proceedings of the Eighth ACM SIGPLAN International
  Conference on Functional Programming}} (Uppsala, Sweden)
  \emph{(\bibinfo{series}{ICFP '03})}. \bibinfo{publisher}{Association for
  Computing Machinery}, \bibinfo{address}{New York, NY, USA},
  \bibinfo{pages}{189--201}.
\newblock
\showISBNx{1-58113-756-7}
\urldef\tempurl%
\url{https://doi.org/10.1145/944705.944723}
\showURL{%
\tempurl}


\bibitem[Wadler(2005)]%
        {Wadler2005}
\bibfield{author}{\bibinfo{person}{Philip Wadler}.}
  \bibinfo{year}{2005}\natexlab{}.
\newblock \showarticletitle{Call-by-Value Is Dual to Call-by-Name - Reloaded}.
  In \bibinfo{booktitle}{\emph{Term Rewriting and Applications, 16th
  International Conference, {RTA} 2005, Nara, Japan, April 19-21, 2005,
  Proceedings}} \emph{(\bibinfo{series}{Lecture Notes in Computer Science},
  Vol.~\bibinfo{volume}{3467})},
  \bibfield{editor}{\bibinfo{person}{J{\"{u}}rgen Giesl}} (Ed.).
  \bibinfo{publisher}{Springer}, \bibinfo{pages}{185--203}.
\newblock
\urldef\tempurl%
\url{https://doi.org/10.1007/978-3-540-32033-3_15}
\showURL{%
\tempurl}


\bibitem[Wright and Felleisen(1994)]%
        {Wright1994Soundness}
\bibfield{author}{\bibinfo{person}{Andrew~K. Wright} {and}
  \bibinfo{person}{Matthias Felleisen}.} \bibinfo{year}{1994}\natexlab{}.
\newblock \showarticletitle{A Syntactic Approach to Type Soundness}.
\newblock \bibinfo{journal}{\emph{Information and Computation}}
  \bibinfo{volume}{115}, \bibinfo{number}{1} (\bibinfo{date}{11}
  \bibinfo{year}{1994}), \bibinfo{pages}{38--94}.
\newblock
\showISSN{0890-5401}
\urldef\tempurl%
\url{https://doi.org/10.1006/inco.1994.1093}
\showDOI{\tempurl}


\bibitem[Zeilberger(2008)]%
        {Zeilberger2008unity}
\bibfield{author}{\bibinfo{person}{Noam Zeilberger}.}
  \bibinfo{year}{2008}\natexlab{}.
\newblock \showarticletitle{On the Unity of Duality}.
\newblock \bibinfo{journal}{\emph{Annals of Pure and Applied Logic}}
  \bibinfo{volume}{153}, \bibinfo{number}{1-3} (\bibinfo{year}{2008}),
  \bibinfo{pages}{66--96}.
\newblock
\urldef\tempurl%
\url{https://doi.org/10.1016/j.apal.2008.01.001}
\showDOI{\tempurl}


\bibitem[Zeilberger(2009)]%
        {Zeilberger2009}
\bibfield{author}{\bibinfo{person}{Noam Zeilberger}.}
  \bibinfo{year}{2009}\natexlab{}.
\newblock \emph{\bibinfo{title}{The Logical Basis of Evaluation Order and
  Pattern-Matching}}.
\newblock \bibinfo{thesistype}{Ph.\,D. Dissertation}. \bibinfo{school}{Carnegie
  Mellon University}, \bibinfo{address}{USA}.
\newblock Advisor(s) Pfenning, Frank and Lee, Peter.
\newblock
\showISBNx{9781109163018}


\bibitem[Zhang et~al\mbox{.}(2016)]%
        {Zhang2016}
\bibfield{author}{\bibinfo{person}{Yizhou Zhang}, \bibinfo{person}{Guido
  Salvaneschi}, \bibinfo{person}{Quinn Beightol}, \bibinfo{person}{Barbara
  Liskov}, {and} \bibinfo{person}{Andrew~C. Myers}.}
  \bibinfo{year}{2016}\natexlab{}.
\newblock \showarticletitle{Accepting Blame for Safe Tunneled Exceptions}. In
  \bibinfo{booktitle}{\emph{Proceedings of the 37th ACM SIGPLAN Conference on
  Programming Language Design and Implementation}} (Santa Barbara, CA, USA)
  \emph{(\bibinfo{series}{PLDI '16})}. \bibinfo{publisher}{Association for
  Computing Machinery}, \bibinfo{address}{New York, NY, USA},
  \bibinfo{pages}{281--295}.
\newblock
\showISBNx{9781450342612}
\urldef\tempurl%
\url{https://doi.org/10.1145/2908080.2908086}
\showDOI{\tempurl}


\end{thebibliography}
